\documentclass[journal]{IEEEtran}

\usepackage[font={small}]{caption}
\usepackage{epsfig,amsmath,amssymb,epsf,amsthm,scalefnt,multirow,subfig,dsfont}
\usepackage{xcolor}
\usepackage{float}
\usepackage{cite}
\usepackage{psfrag}
\usepackage{mathtools}

\newtheorem{theorem}{Theorem}
\newtheorem{remark}{Remark}
\newtheorem{lemma}{Lemma}
\newtheorem*{lemma*}{Lemma}

\newtheorem{definition}{Definition}

\newtheorem{example}{Example}

\def\trace{\mathop{\mathrm{tr}}}
\def\det{\mathop{\mathrm{det}}}
\def \SNR{\operatorname{SNR}}
\def \INR{\operatorname{INR}}

\def \mod{\operatorname{mod}}

\def \P{\operatorname{Pr}}

\def\b0{{\pmb{0}}} 

\newcommand\blfootnote[1]{%
  \begingroup
  \renewcommand\thefootnote{}\footnote{#1}%
  \addtocounter{footnote}{-1}%
  \endgroup}
\allowdisplaybreaks[4]

\begin{document}

\title{Throughput Scaling of Covert Communication over Wireless Adhoc Networks}

 \author{\IEEEauthorblockN{Kang-Hee Cho, Si-Hyeon Lee, and Vincent Y. F. Tan}
}

\maketitle

\begin{abstract}
We consider the problem of covert communication over wireless adhoc networks in which (roughly) $n$ legitimate nodes (LNs) and $n^{\kappa}$ for $0<\kappa<1$ non-communicating warden nodes (WNs) are randomly distributed in a square of unit area. Each legitimate source wants to communicate with its intended destination node while ensuring that every WN is unable to detect the presence of the communication. In this scenario, we study the throughput scaling law. Due the covert communication constraint, the transmit powers are necessarily limited. Under this condition, we introduce a preservation region around each WN. This region serves to prevent transmission from the LNs and to increase   the transmit power of the LNs outside the preservation regions. For the achievability results, multi-hop (MH), hierarchical cooperation (HC), and hybrid HC-MH schemes are utilized with some appropriate modifications. In the proposed MH and hybrid schemes, because the preservation regions may impede communication along direct data paths, the data paths are suitably modified by taking a detour around each preservation region. To avoid the concentration of detours resulting extra relaying burdens, we distribute the detours evenly over a wide region. In the proposed HC scheme, we control the symbol power and the scheduling of distributed multiple-input multiple-output transmission. We also present matching upper bounds on the throughput scaling under the assumption that every active LN consumes the same average transmit power over the time period in which the WNs observe the channel outputs. 
\blfootnote{K.-H. Cho and S.-H Lee are with the Department of Electrical Engineering, Pohang University of Science and Technology, Pohang 37673, South Korea (e-mail: kanghee.cho@postech.ac.kr; sihyeon@postech.ac.kr). 

V. Y. F. Tan is with the Department of Electrical and Computer Engineering, National University of Singapore, Singapore 117583, and also with the Department of Mathematics, National University of Singapore, Singapore 119076  (e-mail: vtan@nus.edu.sg).

The material in this paper will be presented in part at ISIT 2019 in Paris, France in July 2019.}
\end{abstract}

 
\section{Introduction}\label{sec:intro}
Covert communication is the problem of designing protocols to ensure that  communication between legitimate parties is reliable while the presence of such communication is undetectable to an adversary or warden. From the information-theoretic point of view, the fundamental limits of covert communication have been well-studied for point-to-point (P-to-P) scenarios such as additive white Gaussian noise (AWGN) channels \cite{Bash:13}, discrete memoryless channels (DMCs) \cite{Wang:16,Bloch:16}, and several models taking into account channel uncertainty with or without using multiple antennas in the channel model~\cite{deniableuncertainty, undetectable, MIMOAWGN, noncoherent, csit}. The study of covert communication in P-to-P scenarios has been extended to several network information theory models---e.g., multiple access channels \cite{MAC}, broadcast channels \cite{BC}, and random network models in which a multiple number of wardens and interferers (helpers) are distributed randomly \cite{Soltani:18}. In \cite{Soltani:18}, the authors considered a single legitimate source-destination (S-D) pair. In this work, we consider the scenario in which a large number of legitimate S-D pairs communicate covertly without any artificial noise generation. In the covert communication setup, it is known that in the most interesting non-degenerate cases, the maximum throughput over $l$ channel uses scales proportionally to square root of the blocklength $l$ \cite{Bash:13, Wang:16, Bloch:16, BC}---the so-called square-root law; this law has also been observed in several studies on steganography \cite{ste1, ste2}.

In this paper, we consider a more general network setup, namely, wireless adhoc networks where a number of legitimate nodes (LNs) and non-communicating warden nodes (WNs) are distributed randomly in the same network. In the covert communication setting, the communication engineer is required to ensure that communication between the legitimate parties must be undetectable to every WN when each WN observes and uses the channel outputs of a certain blocklength to perform a hypothesis test to determine whether communication takes place. A practical scenario of this setup is a military situation where an army is attempting to infiltrate the enemy's area but the presence of communication is to be hidden from the adversary. As an alternative to characterizing the exact capacity, we follow the traditional approach for wireless adhoc networks \cite{Gupta-Kumar:00}; we study the scaling of the network's capacity of covert communication in the number of LNs. In addition, we observe how the number of WNs and the number of channel uses utilized by the WNs to test the presence of communication are related to the optimal throughput scaling of covert communication. 

In the absence of the covert communication constraint (or WNs), the study of the capacity scaling of wireless adhoc networks has been extensively pursued; see for example \cite{Gupta-Kumar:00,percol, upper, Ozgur:07, arbitrary, regime, Jeon:11}. In particular, in \cite{regime}, the capacity scaling of network with an arbitrary size is derived by using a cutset bound technique and three network schemes: nearest multi-hop (MH) \cite{Gupta-Kumar:00}, hierarchical cooperation (HC) \cite{Ozgur:07}, and hybrid HC-MH \cite{regime} schemes of which the best scheme is decided by the typical nearest signal-to-noise ratio and the path loss exponent. In the presence of the WNs, our network model is somewhat similar to the cognitive network setup in \cite{Jeon:11}, where a primary network and a secondary network share the same network and the secondary network is required to communicate while not affecting the throughput scaling of the primary network. In \cite{Jeon:11}, a preservation region around each primary node is introduced. This region prevents the transmission from secondary nodes to reduce the interference from nearby secondary~nodes, and an MH scheme \cite{Gupta-Kumar:00} is modified to avoid transmission through the preservation regions. 

In this paper, we utilize a preservation region \cite{Jeon:11} around each WN in which the transmission of information from LNs is not permitted. Since the covert communication constraint is more stringent than the constraint given in the secondary network in \cite{Jeon:11}, in this paper, the sizes of preservation regions have to be determined appropriately. In the presence of the preservation regions, we propose three network communication schemes for covert communication that are based on existing schemes in the literature: MH \cite{Gupta-Kumar:00}, HC \cite{Ozgur:07}, and hybrid HC-MH \cite{regime} schemes. These schemes are modified by taking into account the preservation regions and the covert communication constraint. Even if such modifications are made, however, each modified scheme can achieve asymptotically same throughput scaling that the original scheme can achieve when the modified schemes and the original schemes are performed under the same average transmit power constraint. In addition, the modified schemes attain the upper bounds (converses) on the throughput. The scheme that yields the highest throughput amongst these schemes is decided by the number of WNs, the number of channel outputs that each WN observes, and the path loss exponent. These factors determine the operating regime of the network. For the converse part, by using a technique based on the cutset bound, we present matching upper bounds on the throughput scaling under the assumption that every LN outside the preservation regions consumes the same average transmit power over the time period in which the WNs observe the channel outputs. 

The paper is structured as follows. In Section \ref{sec:problem}, we formulate our problem by describing the network model, covert communication constraint, and goal of this paper. In Section \ref{sec:results}, we state the main results on the throughput scalings and give intuitive descriptions for the results. The details on the three proposed schemes are presented in Section \ref{sec:achiev}, and the proofs for cutset bounds are given in Section \ref{sec:conv}. In Section \ref{sec:conclusion}, we conclude this paper.

\section{Problem formulation}\label{sec:problem}

\subsection{Network Model}\label{subsec:network}
We consider a network of a square of unit area where LNs are distributed according to a Poisson point process (p.p.p.) of density $n$. The S-D nodes are chosen in pairs in such a way where each LN is a source or destination such that half of the LNs are sources and the other half are destinations. The communications between the legitimate S-D pairs are required to be covert against non-communicating WNs, which are distributed according to a p.p.p. of density $n^{\kappa}$ for $0<\kappa<1$ in the same network 
\begin{figure}
 \centering
  {
  \includegraphics[width=45mm]{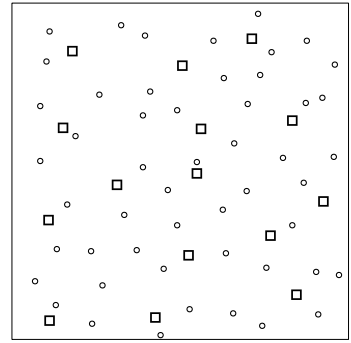}}
  \caption{The network model considered in this paper. The circles represent the LNs and the squares are the WNs.} \label{fig:network}
\end{figure}
(Fig. \ref{fig:network}). The locations of all LNs and WNs are assumed to be fixed during the communication, and every LN knows the locations of other LNs and WNs. We suppose that there are $n_{\mathrm{l}}$ LNs and $n_{\mathrm{w}}$ WNs in the network. 

Each LN has an average transmit power constraint of $P$ and the network is allocated a total bandwidth $B$ around the carrier frequency $f_{\mathrm{c}}\gg B$. Every legitimate source node (LSN) wants to communicate with its legitimate destination node (LDN) at the same rate of $R(n,\kappa)$. The discrete-time baseband-equivalent output at LN $i$ at time $m$ is given as
\begin{align}
Y_i[m]=\sum^{n_{\mathrm{l}}}_{k=1} H_{ik}[m]X_k[m] +N_i[m],
\end{align}
where $X_k[m]$ is the discrete-time baseband-equivalent input at LN $k$, $N_i[m]$ is white circularly symmetric Gaussian noise $\mathcal{CN}(0,N_0)$ at LN $i$, and $H_{ik}[m]$ is the discrete-time baseband-equivalent channel gain between LNs $k$ and $i$ given as
\begin{align}
H_{ik}[m]&=\frac{\sqrt{G}}{d_{ik}^{\alpha/2}}\exp\left(j\theta_{ik}[m] \right),
\label{eqn:H_su}
\end{align}
where $d_{ik}$ is the distance between LNs $k$ and $i$,  $\alpha > 2$ is the path loss exponent in the network, $\theta_{ik}[m]$ is independent and identically distributed (i.i.d.) random phase uniformly distributed over $[0,2\pi]$, and $G=\frac{\lambda^2 G_l}{16\pi^2}$  by Friis' formula, where $G_l$ is the product of the transmit and receive antenna gains. We assume that the legitimate receivers have channel phase information.
Similarly, the discrete-time baseband-equivalent output at WN $i$ at time $m$ is given as 
\begin{align}
Z_i[m]=\sum^{n_{\mathrm{l}}}_{k=1} H'_{ik}[m]X_k[m] +N_i'[m],
\end{align}
where $N_i'[m]$ is white circularly symmetric Gaussian noise $\mathcal{CN}(0,N_0)$ at WN $i$, and $H'_{ik}[m]$ is the channel gain between LN $k$ and WN $i$, which is defined similarly to \eqref{eqn:H_su}. Henceforth, we will omit the time index for brevity. We assume that there exists a sufficiently long secret key between every LN S-D pair.
\begin{remark}
Fix $\epsilon > 0$. According to \cite[Lemma 1]{Jeon:11}, $n_{\mathrm{l}}$ belongs to the interval $((1-\epsilon)n, (1+\epsilon)n)$ and $n_{\mathrm{w}}$ belongs to the interval $((1-\epsilon)n^{\kappa}, (1+\epsilon)n^{\kappa})$ w.h.p.\footnote{The term w.h.p.\ means that the probability of the event tends to one as $n$ tends to infinity.}
\end{remark}

\subsection{Covertness Constraint}\label{subsec:covertness}
In this subsection, the covert communication constraint is presented. Each WN observes and utilizes the channel outputs over $l$ channel uses to test whether LNs are communicating. The integer $l$ can be regarded as a memory constraint at each WN and the starting point of the window can be arbitrarily chosen over the whole duration in which communication takes place. We assume that $l$ is sufficiently large, and it can also tend to infinity at a certain rate as $n$ tends to infinity. Then, the covert communication constraint is expressed as follows:\footnote{It is known from Pinsker's inequality that the hypothesis test by a WN is not much better than a blind test when the relative entropy is small; see \cite{Cover:2006}. In this paper, the relative entropy is defined using the natural logarithm.}
\begin{align}\label{eqn:covertness}
D\left(\left. \hat{Q}_{Z_i^l} \right \| Q_{N_i'}^{\times l}\right)\leq \delta, \quad  i=1,2,\ldots,n_{\mathrm{w}},
\end{align} 
where $\delta >0$ and $Q_{N_i'}$ is the distribution of $N_i'$ when no communication occurs (i.e., the transmitter remains silent), $Q_{N_i'}^{\times l}$ is the $l$-fold product distribution of $Q_{N_i'}$, and $\hat{Q}_{Z_i^l}$ is the distribution of $Z_i^l$ (i.e., the induced output distribution) when the transmitter is active and a message is sent.

\subsection{Goal of This Paper}\label{subsec:goal}
We allow some LN pairs not to send their own message as long as the fraction of such pairs is {\em vanishing}.\footnote{By {\em vanishing}, we mean that the fraction of the LN pairs that do not send their own message tends to zero as $n$ tends to infinity.}  Let $\epsilon(n,\kappa) > 0$ denote this outage fraction of the LNs as a function of $n$ and $\kappa$ and we say that the LN pairs that do not send their own messages are {\em in outage}. The LNs in outage can, however, cooperate with other LNs to enable communication. Under this assumption, the following is the formal definition of an achievable throughput per LN.
\begin{definition}
 A throughput of $R(n,\kappa)$ per LN is said to be achievable under the covertness constraint if at least $1-\epsilon(n,\kappa)$ fraction of LSNs can transmit with a rate of $R(n,\kappa)$ to their LDNs w.h.p.\ while satisfying the covertness constraint \eqref{eqn:covertness}. We note that $\epsilon(n, \kappa)$ vanishes as $n \to \infty$.
\end{definition}
The aggregate throughput $T(n,\kappa)$ is defined as $\frac{n_{\mathrm{l}}R(n,\kappa)}{2}$. Our goal is to characterize the aggregate throughput scaling law in the presence of the covertness contraint.

\section{Main results}\label{sec:results}
In this section, we present achievable throughput scalings in Theorems \ref{thm:achv1} and \ref{thm:achv2}. In Theorems \ref{thm:conv1} and \ref{thm:conv2}, we also present upper bounds on the aggregate throughput scalings under the assumption that every active LN uses the {\em same} average transmit power over an arbitrary window of $l$ channel uses when the covertness constraint results in a more stringent average transmit power constraint compared to the average transmit power constraint of $P$ at each LN (as described in Section II-A). Our achievability schemes can achieve the upper bounds asymptotically. Proofs of Theorems \ref{thm:achv1} and \ref{thm:achv2} are presented in Section \ref{sec:achiev}, and proofs of Theorems \ref{thm:conv1} and \ref{thm:conv2} are in Section \ref{sec:conv}.

\begin{figure}
 \centering
  {
  \includegraphics[width=90mm]{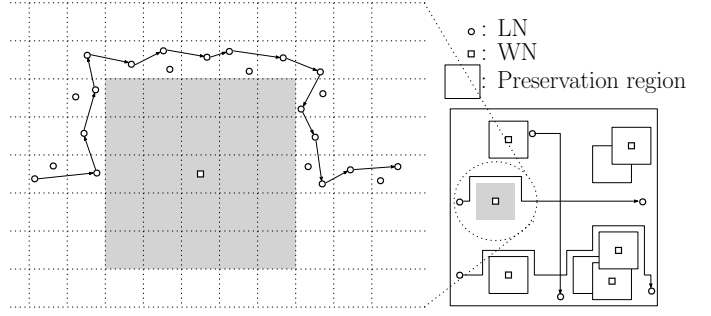}}
  \caption{Example of a MH scheme with detouring paths in the presence of the preservation regions. } \label{fig:detouringMHe}
\end{figure}

 For the achievability, since the covertness requirement constrains the interference power at each WN, the transmit powers of the LNs are also constrained, similarly to the AWGN case \cite{Wang:16}. Hence, we introduce a preservation region around each WN where the transmission from LNs is not permitted to increase the transmit power of the active LNs. In the presence of the covertness constraint and the preservation regions, we propose modified versions of three existing schemes: MH \cite{Gupta-Kumar:00}; HC \cite{Ozgur:07}; hybrid HC-MH \cite{regime}. In the proposed MH and hybrid schemes, because the preservation regions may block the direct paths of the S-D pairs, we modify these paths so that they take a detour around the preservation regions as in Fig. \ref{fig:detouringMHe}. In addition, to avoid the fact that the relaying burden is concentrated at some LNs thus causing throughput degradation, we distribute the detouring paths appropriately. In the proposed HC scheme, we control the transmit power and the scheduling of multiple-input multiple-output (MIMO) transmissions appropriately to satisfy the covertness constraint while retaining the performance of the original HC scheme. 

\begin{theorem}\label{thm:achv1}
Fix $\epsilon>0$. If $\frac{n^{(\frac{1}{2}-\frac{\kappa}{2})(\alpha-2)}}{\sqrt{l}}$ vanishes as $n$ tends to infinity, the following aggregate throughput $T(n,\kappa)$ is achievable under the covertness constraint w.h.p.:
\begin{align*}
T(n,\kappa)=
\begin{cases}
\Omega \left(n^{2-\frac{\alpha}{2}-\epsilon} \cdot \frac{n^{(\frac{1}{2}-\frac{\kappa}{2})(\alpha-2)}}{\sqrt{l}} \right), &\mbox{if } 2 < \alpha \leq 3 \\
\Omega \left( n^{\frac{1}{2}-\epsilon} \cdot \frac{n^{(\frac{1}{2}-\frac{\kappa}{2})(\alpha-2)}}{\sqrt{l}} \right), &\mbox{if } \alpha > 3
\end{cases}
\end{align*}
\end{theorem}
The proposed HC and MH schemes are used for $2<\alpha \leq 3$ and $\alpha > 3$, respectively. 

\begin{theorem}\label{thm:achv2}
Fix $\epsilon>0$. If $\frac{n^{(\frac{1}{2}-\frac{\kappa}{2})(\alpha-2)}}{\sqrt{l}}$ does not vanish as $n$ tends to infinity,\footnote{A positive sequence $g(n)$ does not vanish as $n$ tends to infinity if $\liminf_{n\to\infty}g(n)>0$.} the following aggregate throughput $T(n,\kappa)$ is achievable under the covertness constraint w.h.p.:
\begin{align*}
T(n,\kappa)=
\begin{cases}
\Omega \left(n^{2-\frac{\alpha}{2}-\epsilon} \cdot \frac{n^{(\frac{1}{2}-\frac{\kappa}{2})(\alpha-2)}}{\sqrt{l}} \right), &\mbox{if } 2 < \alpha \leq 3 \\
\Omega \left( n^{\frac{1}{2}-\epsilon} \left( \frac{n^{(\frac{1}{2}-\frac{\kappa}{2})(\alpha-2)}}{\sqrt{l}}\right)^{\frac{1}{\alpha-2}} \right), &\mbox{if } \alpha > 3
\end{cases}
\end{align*}
\end{theorem}
The proposed HC and hybrid schemes are used for $2<\alpha \leq 3$ and $\alpha > 3$, respectively.

Roughly speaking, in our three proposed schemes, the covertness constraint and the preservation regions result in the average transmit power constraint over an arbitrary window of $l$ channel uses at each LN scaling  with the order of $\frac{n^{(\frac{1}{2}-\frac{\kappa}{2})(\alpha-2)}}{n^{\alpha/2}\sqrt{l}}$. Then, the received signal-to-noise ratio ($\SNR$) between the typical nearest neighbor LNs, i.e., between two LNs distance of $\sqrt{1/n}$, is given as $\frac{n^{(\frac{1}{2}-\frac{\kappa}{2})(\alpha-2)}}{\sqrt{l}}$. The $\sqrt{l}$ in the denominator demonstrates the similarity to the square-root law in covert communication \cite{Wang:16}. In addition, because large $\alpha$ results in more attenuated interference signal at each WN, the LNs can retain higher transmit power when $\alpha$ is large. This short-range $\SNR$ term and the path loss exponent $\alpha$ are the two parameters that constitute the condition of the network and have to be carefully considered to select an appropriate high-throughput scheme (amongst the three proposed ones) for a certain network model. 

When the term $\frac{n^{(\frac{1}{2}-\frac{\kappa}{2})(\alpha-2)}}{\sqrt{l}}$ is vanishing, it plays the role of a linear scaling factor of the throughput in the proposed MH and HC schemes. This is because in the proposed MH scheme, the rate in a hop is propotional to the received $\SNR$, which vanishes (by assumption). In the proposed HC scheme, the MIMO gain is propotional to the transmit power regardless of whether the short-range $\SNR$ term is vanishing.
This yields the throughput scalings given by the product of two terms, where one is the same as the throughput scalings of the extended network \cite{Ozgur:07} and the other term, namely $\frac{n^{(\frac{1}{2}-\frac{\kappa}{2})(\alpha-2)}}{\sqrt{l}}$ is a linear scaling factor. When the term $\frac{n^{(\frac{1}{2}-\frac{\kappa}{2})(\alpha-2)}}{\sqrt{l}}$ is non-vanishing, because the proposed MH scheme cannot efficiently utilize the non-vanishing received $\SNR$ in a hop, the hybrid scheme that can efficiently utilize the received $\SNR$ outperforms the proposed MH scheme.

The following theorems present upper bounds on the aggregate throughput scalings under the assumption that every active LN uses the same average transmit power over an arbitrary window of $l$ channel uses. 

\begin{theorem}\label{thm:conv1}
Fix $\epsilon>0$. When every active LN uses the same average transmit power over an arbitrary window of $l$ channel uses and $\frac{n^{(\frac{1}{2}-\frac{\kappa}{2})(\alpha-2)}}{\sqrt{l}}$ vanishes as $n$ tends to infinity, the aggregate throughput $T(n,\kappa)$ under the covertness constraint is upper bounded as
\begin{align*}
T(n,\kappa) =
\begin{cases}
O \left(n^{2-\frac{\alpha}{2}+\epsilon} \cdot \frac{n^{(\frac{1}{2}-\frac{\kappa}{2})(\alpha-2)}}{\sqrt{l}} \right), &\mbox{if } 2 < \alpha \leq 3 \\
O \left( n^{\frac{1}{2}+\epsilon} \cdot \frac{n^{(\frac{1}{2}-\frac{\kappa}{2})(\alpha-2)}}{\sqrt{l}} \right), &\mbox{if } \alpha > 3
\end{cases}
\end{align*}
\end{theorem}

\begin{theorem}\label{thm:conv2}
Fix $\epsilon>0$. When every active LN uses the same average transmit power over an arbitrary window of $l$ channel uses and $\frac{n^{(\frac{1}{2}-\frac{\kappa}{2})(\alpha-2)}}{\sqrt{l}}$ is does not vanish as $n$ tends to infinity, the aggregate throughput $T(n,\kappa)$ under the covertness constraint is upper bounded as
\begin{align*}
T(n,\kappa)=
\begin{cases}
O \left(n^{2-\frac{\alpha}{2}+\epsilon} \cdot \frac{n^{(\frac{1}{2}-\frac{\kappa}{2})(\alpha-2)}}{\sqrt{l}} \right), &\mbox{if } 2 < \alpha \leq 3 \\
O \left( n^{\frac{1}{2}+\epsilon} \left( \frac{n^{(\frac{1}{2}-\frac{\kappa}{2})(\alpha-2)}}{\sqrt{l}}\right)^{\frac{1}{\alpha-2}} \right), &\mbox{if } \alpha > 3
\end{cases}
\end{align*}
\end{theorem}
In the converse proof, the preservation region is also introduced, and a cutset bound technique similar to \cite[Section III]{regime} is used to derive Theorems \ref{thm:conv1} and \ref{thm:conv2}. 

\begin{figure}
 \centering
  {
  \includegraphics[width=85mm]{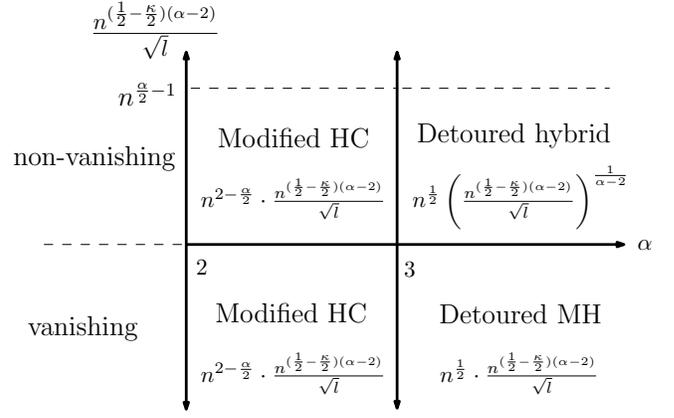}}
  \caption{An illustration of the four operating regimes depending on whether $\frac{n^{(\frac{1}{2}-\frac{\kappa}{2})(\alpha-2)}}{\sqrt{l}}$ vanishes and the range of $\alpha$. In each regime, the optimal scheme (under the assumption on the transmit power of LNs in the converse) and its optimal throughput scaling is shown. } \label{fig:theorems}
\end{figure}

The networks can be partitioned into four operating regimes according to the two the following parameters: the short-range $\SNR$ $\frac{n^{(\frac{1}{2}-\frac{\kappa}{2})(\alpha-2)}}{\sqrt{l}}$ and the path loss exponent $\alpha$ as shown in Fig. \ref{fig:theorems}. Under the assumption that every active LN has the same average transmit power constraint over an arbitrary window of $l$ channel uses, the optimal scheme and its optimal throughput scaling in each regime are illustrated in Fig. \ref{fig:theorems}. 

The partitioning of the operating regimes of our network is similar but slightly different from the operating regimes in \cite{regime}. The term $\frac{n^{(\frac{1}{2}-\frac{\kappa}{2})(\alpha-2)}}{\sqrt{l}}$ plays a role similarly to that of the received $\SNR$ in the typical nearest neighbor distance in \cite{regime}; the latter being denoted in \cite{regime} as $\SNR_s$. 
Since the covertness constraint makes it impossible to retain sufficient long-range $\SNR$ regardless of the values of $\kappa$ and $l$, there is no regime in our network that corresponds to Regime I (retaining sufficient long-range $\SNR$) in \cite{regime} in which the HC scheme is the optimal scheme. 
In our case with $2<\alpha \leq 3$ is similar to Regime II (insufficient long-range $\SNR$ but path-loss is weak, i.e., $2<\alpha \leq 3$) in \cite{regime} in which the bursty HC scheme is optimal. 
Regimes III (path-loss is strong, i.e., $\alpha > 3$, but $\SNR_s$ is sufficient and long-range $\SNR$ is not sufficient) in \cite{regime} in which the hybrid scheme is optimal corresponds to our setting with $\alpha > 3$ and non-vanishing $\frac{n^{(\frac{1}{2}-\frac{\kappa}{2})(\alpha-2)}}{\sqrt{l}}$.  
Our setting with $\alpha > 3$ and vanishing $\frac{n^{(\frac{1}{2}-\frac{\kappa}{2})(\alpha-2)}}{\sqrt{l}}$ is analogous to Regime IV ($\alpha > 3$ and $\SNR_s \ll 0$) in \cite{regime} in which the MH scheme is optimal.

\section{Achievability}\label{sec:achiev}
In this section, we describe the nature of the preservation regions in Section \ref{subsec:preserv} and also the three proposed schemes, namely detoured MH, modified HC, and detoured hybrid schemes in Sections \ref{subsec:MH}, \ref{subsec:HC}, and \ref{subsec:hybrid}, respectively.

\subsection{Preservation Region}\label{subsec:preserv}
In the presence of the covertness constraint, the power of the interference at each WN is limited and thus the transmit powers of the LNs are also constrained. Hence, we introduce a preservation region \cite{Jeon:11}, which is shaped as a square around each WN in which transmission from LNs is not permitted. Each preservation region ensures that the interference to each WN from the nearby LNs is not too strong. This makes it possible to increase the transmit powers of the LNs outside the preservation regions. Intuitively, when the received $\SNR$ between the LNs is low compared to the noise level, increasing the transmit power of the LNs leads to a significant increase in the throughput. From this intuition, we set the width of the preservation region $b$ to be\footnote{For the scenario in which the outage fraction vanishes, $\gamma > \kappa/2$ is necessary.}
\begin{align}
b = \Theta(n^{-\gamma}\sqrt{\log n}),~ \gamma \leq \frac{1}{2},
\end{align}
where a preservation region with $b$ such that $\gamma = \frac{1}{2}$ corresponds to a square of nine cells each with size $\frac{2\log n}{n}$ that contains a WN in the center cell, and $b$ with $\gamma < \frac{1}{2}$ is an integer multiple of $\sqrt{\frac{2\log n}{n}}$. The preservation region exponent $\gamma$ has to be chosen in a way that maximizes the throughput while ensuring that the resulting outage fraction vanishes as $n$ tends to infinity. The proofs for vanishing outage fraction of the three proposed schemes are presented in Appendix \ref{app:outage}.

\subsection{Detoured MH Scheme}\label{subsec:MH}
We propose a modified MH scheme, termed {\em detoured} MH scheme that avoids the preservation regions by taking a detour around each of the preservation regions. The detoured MH scheme follows the structure of the original MH scheme described as follows. 
\subsubsection{Original MH Scheme}\label{subsubsec:originalMH}
The summary of the original MH scheme \cite{Gupta-Kumar:00, Jeon:11} is given as follows: 
\begin{itemize}
\item Divide the network into square cells with area of $\frac{2\log n}{n}$.
\item Every LN uses the same power for transmission.

\begin{figure}
 \centering
  {
  \includegraphics[width=45mm]{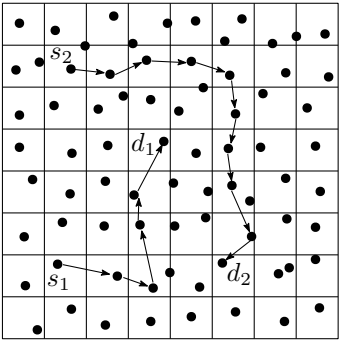}}
  \caption{Example of the original MH scheme. Each data path consists of a horizontal data path and a vertical data path.} \label{fig:originalMH}
\end{figure}

\item The data packet of each LSN is delivered to the LDN by sequentially transmitting to a LN in the adjacent cells first along the horizontal data path (HDP) and then along the vertical data path (VDP) as shown in Fig. \ref{fig:originalMH}.
\item Use a 9-TDMA scheme for the communication between adjacent cells. This means that for each group of a square of nine cells, only a single cell in the same position in each group is active for each time slot. 
\end{itemize}

The following lemma presents lower and upper bounds on the number of LNs in each cell with area of $\frac{2 \log n}{n}$. This lemma is used in some proofs throughout this paper.
\begin{lemma}\label{lem:num_cell}
When the area of a cell is $\frac{2\log n}{n}$, there is at least one LN and are at most $4\log n$ LNs in each cell w.h.p.
\end{lemma}
\begin{proof}
We first record the following fact. For a Poisson random variable $X$ with rate $\lambda$, the following inequality holds:
\begin{align}\label{eqn:inequality}
\P(X \geq k) \leq \frac{e^{-\lambda}(e\lambda)^k}{k^k},~\mbox{for}~k>\lambda.
\end{align} 
We refer to \cite{percol} for the proof of the above inequality.

Let $A$ be the event that there is at least one LN in each cell. Also let $B$ be the event that there are at most $4\log n$ LNs in each cell. If we define $N_k$ as the number of LNs in the $k^{\mathrm{th}}$ cell, then each $N_k$ for $k=1,2,...,\frac{n}{2\log n}$ follows a p.p.p. of density $n$, and $N_1$, $N_2,...$ are mutually independent. Thus, each of the $N_k$ is equivalent to a Poisson random variable with rate $2\log n$, and hence
\begin{align}
\P(A) &= 1 - \P(A^c) \\
&= 1 - \P\left(\bigcup_{k=1}^{\frac{n}{2\log n}}\{ N_k = 0 \}\right) \\ \label{eqn:uni1}
& \geq 1 - \frac{n}{2\log n}\P\left(N_1 = 0\right) \\ \label{eqn:Poi1}
& = 1 - \frac{n}{2\log n}e^{-2\log n},
\end{align}   
where \eqref{eqn:uni1} is due to the union bound and the mutual independence of the cells; \eqref{eqn:Poi1} follows from the fact that $N_1$ is a Poisson random variable with rate $2\log n$. The last term tends to $1$ as $n$ tends to infinity. 

Next, 
\begin{align}
\P(B) &= 1 - \P(B^c) \\
&= 1 - \P\left(\bigcup_{k=1}^{\frac{n}{2\log n}}\{ N_k > 4\log n \}\right) \\ \label{eqn:uni2}
& \geq 1 - \frac{n}{2\log n}\P( N_1 > 4\log n ) \\ \label{eqn:Poi2}
& \geq 1 - \frac{n}{2\log n}e^{-2\log n}\left( \frac{e}{2}\right)^{4\log n},
\end{align}
where \eqref{eqn:uni2} is because of the union bound and the mutual independence of the cells; \eqref{eqn:Poi2} is from \eqref{eqn:inequality}. The last term goes to $1$ as $n$ tends to infinity.  
\end{proof}


\subsubsection{A Detouring Method}\label{subsubsec:detouredMH}
Because the preservation regions may block the HDPs and VDPs of some S-D pairs, such paths must be modified so that they do not overlap with the preservation regions. Thus, we modify these paths by taking a detour around the preservation regions. However, the detouring paths generate extra relaying burdens to some cells, and this causes significant throughput degradation when the detouring paths are concentrated at some cells. 
\begin{figure}
 \centering
  {
  \includegraphics[width=90mm]{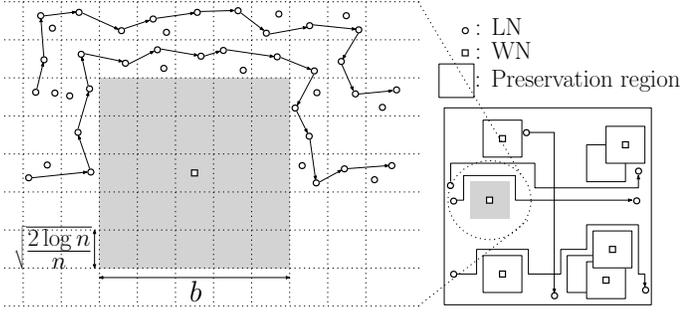}}
  \caption{Examples of modified data paths that detour the preservation regions. They are distributed to a wide region with width of the order of $b$.} \label{fig:seperate}
\end{figure}
Our solution is to distribute the detouring paths to a wide region with width of the order $b$ to avoid the concentration of the relaying burdens to some cells as shown in Fig. \ref{fig:seperate}.

Let us describe the detouring method in more detail. First we note that there can be some preservation regions that overlap with each other or are close to each other because the WNs are distributed randomly. Therefore, the detouring paths have to be chosen appropriately. Our proposed detouring method is as follows. First, define the {\em distance} between two preservation regions as the shortest distance of any two points in each preservation region. In addition, we say a set of preservation regions comprise a {\em cluster} if for each preservation region in the cluster there is at least one preservation region in the same cluster with a distance less than $b$. 
\begin{figure}
 \centering
  {
  \includegraphics[width=50mm]{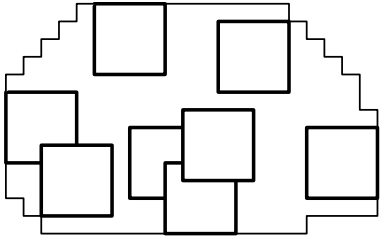}}
  \caption{Example of an expanded preservation region. The white squares are the preservation regions in a cluster.} \label{fig:expanded}
\end{figure}
Finally, we define the {\em expanded preservation region} as the set of cells that partially or completely belong to the convex closure of the preservation regions in a cluster (Fig. \ref{fig:expanded}). Then, the following lemma holds. 
\begin{lemma}\label{lem:num_pre}
If $\gamma > \kappa/2$, there are at most $4\kappa \log n$ preservation regions in an expanded preservation region w.h.p. 
\end{lemma}

\begin{proof}
According to Lemma \ref{lem:num_cell}, there are at most $4\kappa \log n$ WNs in a square of area $\frac{2\log n^{\kappa}}{n^{\kappa}}$. Consider an expanded preservation region and assume there are $4\kappa \log n + 1$ preservation regions in the expanded preservation region. Then, the size of the expanded preservation region can be as large as $\Theta \left( \frac{(\log n)^2}{n^{2\gamma}} \right)$ which corresponds to the case when the preservation regions in the expanded preservation region do not overlap with each other. This is asymptotically smaller than $\frac{2\log n^{\kappa}}{n^{\kappa}}$. This contradicts Lemma \ref{lem:num_cell}.
\end{proof}

Transmission is not allowed in the expanded preservation regions. We modify the HDPs and VDPs that overlap with some expanded preservation regions to detour the expanded preservation regions. The following lemma presents an upper bound on the relaying burden at each cell when our proposed detouring method is used.

\begin{lemma}\label{lem:detouringmethod}
If $\gamma > \kappa/2$, there exists a detouring method in which each cell needs to carry at most $\Theta(\sqrt{n}(\log n)^{3/2})$ data paths w.h.p.
\end{lemma}

\begin{proof}
\begin{figure}
 \centering
  {
  \includegraphics[width=70mm]{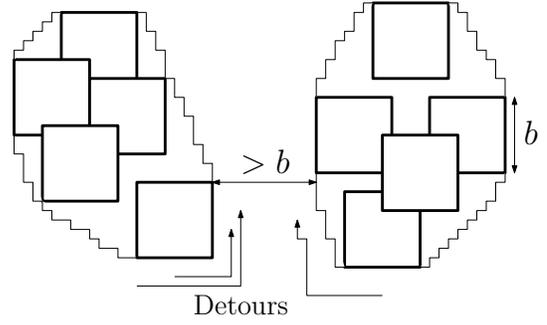}}
  \caption{Example of a passable region between two expanded preservation regions. The width of the passable region is longer than $b$.} \label{fig:tunnel}
\end{figure}
According to the definition of a cluster of preservation regions and the convex closure, there exists a passable region of which width is longer than $b$ between any two expanded preservation regions as shown in Fig. \ref{fig:tunnel}. The detouring paths are distributed into the areas that surround the expanded preservation regions and have width of $b/2$ as follows:
\begin{itemize}
\begin{figure}
 \centering
  {
  \includegraphics[width=70mm]{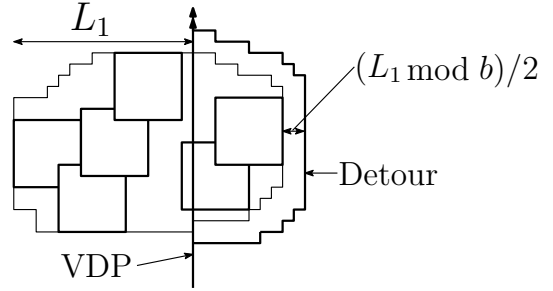}}
  \caption{Example of a detoured VDP.} \label{fig:detouringmethod}
\end{figure}
\item Refer to Fig. \ref{fig:detouringmethod}. Consider a VDP crossing an expanded preservation region. The case of the HDP can be considered in a similar way. To choose a detouring path, first calculate the length between the line of VDP and the farmost line from the VDP that is parallel with the VDP and overlaps with the edge of the expanded preservation region ($L_1$ in Fig. \ref{fig:detouringmethod}). 
\item On the opposite side, take a detour so that the interval between the detour and the border of the expanded preservation region is $(L_1 \mod b)/2$. Using this method, the detouring paths are distributed into areas that have width approximately $b/2$.
\begin{figure}[t]
 \centering
  {
  \includegraphics[width=40mm]{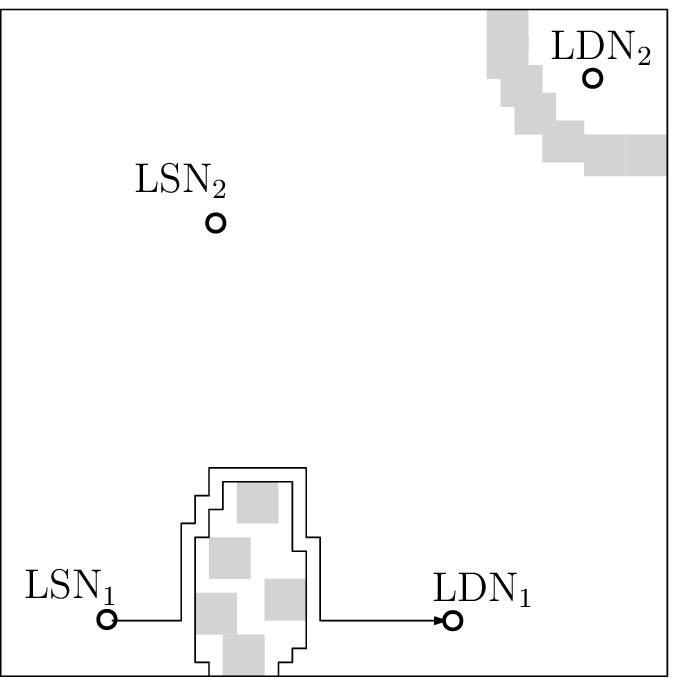}}
  \caption{The white circles are LNs and the gray squares are WNs. There is no possible data path between $\mbox{LSN}_2$ and $\mbox{LDN}_2$.} \label{fig:case}
\end{figure}
\item If a VDP cannot take a detour at one side, we will take a detour on the other side. Furthermore, if there is no possible data path between a LSN and the corresponding LDN, we regard the pair to be in outage (Fig. \ref{fig:case}).   
\end{itemize}
Because there are at most $4\kappa \log n$ preservation regions in a expanded preservation region, and there are at most $4 \log n$ LNs in a cell, the largest possible number of LNs whose VDPs overlap with a certain expanded preservation region is of the order $\Theta(n^{-\gamma}\sqrt{\log n}\cdot \log n \cdot n) = \Theta(n^{1-\gamma}(\log n)^{3/2})$. By distributing these relaying burdens into a region with width of the order of $\Theta(n^{-\gamma}\sqrt{\log n})$, this lemma is proved. 
\end{proof}

As shown in Lemma \ref{lem:detouringmethod}, the preservation region exponent $\gamma$ does not affect the scaling of the relaying burden of each cell. Thus, we enlarge the preservation region up to a certain level to increase the transmit powers of the LNs while ensuring that the outage fraction vanishes as $n$ tends to infinity. As will be proved in Appendix \ref{app:outage}, even if $\gamma = \frac{\kappa}{2} + \epsilon$ where $\epsilon$ is an arbitrarily small constant, the outage fraction vanishes. Thus, in the following, we simply set $\gamma=\frac{\kappa}{2}$ to ensure that the transmit powers of LNs are maximized for notational brevity.

\subsubsection{Codebook Generation}\label{subsubsec:codebook}
Each LN uses a codebook with codewords generated each in an i.i.d.\ manner. The components of the codewords are themselves i.i.d.\ complex Gaussian with zero mean and some appropriately chosen variance $P_{\mathrm{MH}}$, i.e., $\mathcal{CN}(0,P_{\mathrm{MH}})$. The secret key, which is assumed to be sufficiently sufficiently long, is utilized to construct the random codebook (which is known to all parties). We discuss the choice of $P_{\mathrm{MH}}$ in Section \ref{subsubsec:pMH}.

\subsubsection{Rate in a Hop}\label{subsubsec:rate}
We present an achievable rate of the communication between two adjacent cells based on the alloted power $P_{\mathrm{MH}}$. 
Let $\SNR_{\mathrm{hop}}$ be the $\SNR$ of the communication protocol between two adjacent cells. Then, because the longest distance between the two LNs in two neighboring cells is $\sqrt{\frac{10\log n}{n}}$, we have
\begin{align}\label{eqn:hoplow}
\SNR_{\mathrm{hop}} \geq \frac{GP_{\mathrm{MH}}}{N_0} \left( \sqrt{\frac{10\log n}{n}} \right)^{-\alpha}.
\end{align} 
Recall that $G$ is from the antenna gain given in \eqref{eqn:H_su}. 

The following lemma plays a key role in deriving an achievable rate in a hop and is useful throughout this paper.  

\begin{lemma}\label{lem:insn}
Let $\INR_{\mathrm{MH}}$ be the interference-to-noise ratio ($\INR$) at the LNs. Then, if $\alpha > 2$, each LN satisfies
\begin{align}
\INR_{\mathrm{MH}} \leq K_{\mathrm{I}} \SNR_{\mathrm{hop}},
\end{align}
where $K_{\mathrm{I}}$ is a constant independent of $\alpha$, $P_{\mathrm{MH}}$, and $n$. 
\end{lemma}
\begin{proof}
To aid the reader's understanding of this proof, s/he is referred to Fig. \ref{fig:inrbound}. Based on the 9-TDMA scheme, when the communication between the cell marked with $S$ and the cell marked with $D$ occurs, the gray cells except the cell marked with $S$ are the interfering cells. These interfering cells are grouped into the interfering groups $F_1$, $F_2,...$ $F_i,...$ such that $F_i$ consists of $8i$ interfering cells and the LNs in $F_i$ have distance longer than $(3i-2)\sqrt{\frac{2\log n}{n}}$ from the LNs in the cell marked with $D$. Then, the $\INR_{\mathrm{MH}}$ of the receiver LN in the cell marked with $D$ in Fig. \ref{fig:inrbound} is upper bounded as follows: 
\begin{figure}
 \centering
  {
  \includegraphics[width=50mm]{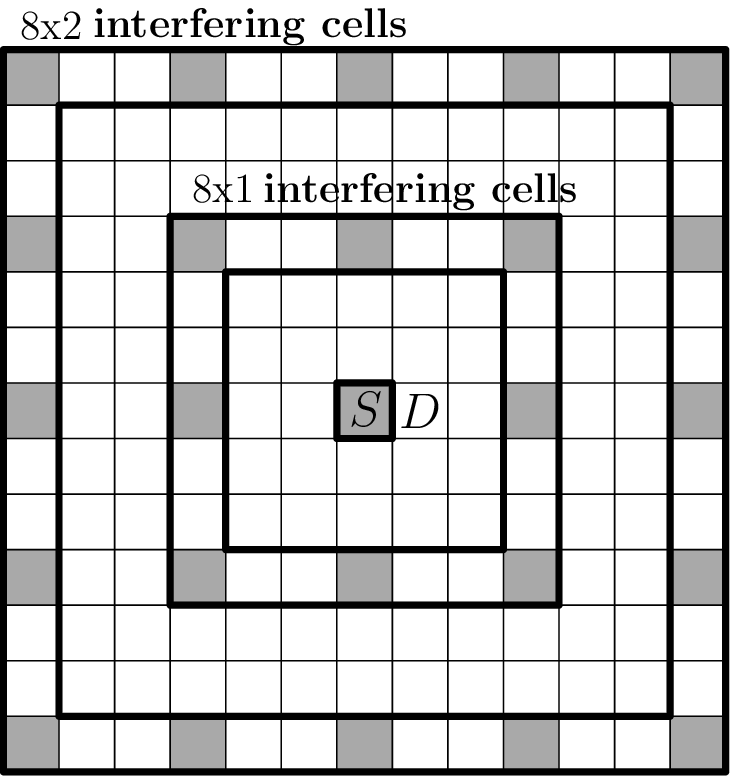}}
  \caption{9-TDMA scheme and interfering cells. The cell marked with $S$ is the cell that a transmitter is in and the cell marked with $D$ is the cell that the corresponding receiver is in. The gray cells except the cell $S$ are the interfering cells that form several interfering groups.} \label{fig:inrbound}
\end{figure}
\begin{align}
 \INR_{\mathrm{MH}} &< \frac{GP_{\mathrm{tx}}}{N_0}\sum_{i=1}^n8i\left((3i-2)\sqrt{\frac{2\log n}{n}}\right)^{-\alpha}\\
&= \frac{GP_{\mathrm{tx}}}{N_0}\left(\sqrt{\frac{2\log n}{n}}\right)^{-\alpha}\sum_{i=1}^n8i(3i-2)^{-\alpha}\\ \label{eqn:ibound1}
&= k\frac{GP_{\mathrm{tx}}}{N_0}\left(\sqrt{\frac{2\log n}{n}}\right)^{-\alpha} \\ \label{eqn:ibound2}
&\leq K_{\mathrm{I}} \SNR_{\mathrm{hop}}, 
\end{align}
where $k$ and $K_{\mathrm{I}}$ are constants independent of $\alpha$, $\delta$, $l$ and $n$, \eqref{eqn:ibound1} is because $\sum_{i=1}^{\infty}8i(3i-2)^{-\alpha}$ converges when $\alpha > 2$, and \eqref{eqn:ibound2} is due to the inequality in \eqref{eqn:hoplow}. 
\end{proof}

Based on the above lemma, we now derive a lower bound on the rate of a hop.

\begin{lemma}\label{lem:cellrate}
The rate $R_{\mathrm{hop}}$ that every cell can achieve in a hop is lower bounded as follows:
\begin{align}
R_{\mathrm{hop}} \geq \frac{1}{9} \log \left(1 + \frac{c  P_{\mathrm{MH}}\left( \sqrt{\frac{2\log n}{n}} \right)^{-\alpha}}{N_0 + c K_I  P_{\mathrm{MH}} \left( \sqrt{\frac{2\log n}{n}} \right)^{-\alpha}}\right) ,
\end{align}
where $c$ and $K_I$ are constants independent of $\alpha$, $P_{\mathrm{MH}}$, and $n$. 
\end{lemma}

\begin{proof}
Based on the 9-TDMA scheme, by using the codebook as desribed in Section \ref{subsubsec:codebook},
\begin{align}
R_{\mathrm{hop}} &= \frac{1}{9} \log \left(1 + \frac{N_0 \SNR_{\mathrm{hop}}}{N_0 + N_0 \INR_{\mathrm{l}}}\right) \\ \label{eqn:Rhop1}
&\geq \frac{1}{9} \log \left(1 + \frac{N_0 \SNR_{\mathrm{hop}}}{N_0 + N_0 K_{\mathrm{I}}\SNR_{\mathrm{hop}}}\right) \\ \label{eqn:Rhop2}
&\geq \frac{1}{9} \log \left(1 + \frac{GP_{\mathrm{MH}}\left( \sqrt{\frac{10\log n}{n}} \right)^{-\alpha}}{N_0 + K_{\mathrm{I}}GP_{\mathrm{MH}} \left( \sqrt{\frac{10\log n}{n}} \right)^{-\alpha}}\right) \\
&= \frac{1}{9} \log \left(1 + \frac{c  P_{\mathrm{MH}}\left( \sqrt{\frac{2\log n}{n}} \right)^{-\alpha}}{N_0 + c K_I  P_{\mathrm{MH}} \left( \sqrt{\frac{2\log n}{n}} \right)^{-\alpha}}\right),
\end{align}
where \eqref{eqn:Rhop1} is due to Lemma \ref{lem:insn}; \eqref{eqn:Rhop2} is from \eqref{eqn:hoplow}, and $c$ and $K_I$ are constants independent of $\alpha$, $P_{\mathrm{MH}}$, and $n$. 
\end{proof}

\subsubsection{Transmit Power $P_{\mathrm{MH}}$}\label{subsubsec:pMH}
We now derive an available transmit power of the LNs that makes it possible to achieve the desired throughput scaling while satisfying the covertness constraint in \eqref{eqn:covertness}. Because each LN uses a codebook generated according to the distribution $\mathcal{CN}(0,P_{\mathrm{MH}})$, the covertness measure in \eqref{eqn:covertness} at WN $i$ is upper bounded as follows:
\begin{align} 
D\left(\left. \hat{Q}_{Z_i^l} \right \| Q_{N_i'}^{\times l}\right) \label{eqn:first}
&= \sum_{t=1}^l D\left(\left. \hat{Q}_{Z_{i,t}} \right \| Q_{N_i'}\right) \\
&= \sum_{t=1}^l D\left (\left. \mathcal{CN}(0,I_{i,t} + N_0)\right \| \mathcal{CN}(0,N_0) \right ) \\
&\leq l\cdot D\left (\left. \mathcal{CN}(0,I_i + N_0)\right \| \mathcal{CN}(0,N_0) \right ) \\
&= l\cdot \left(\frac{I_i}{N_0} -\log \frac{I_i+N_0}{N_0}  \right)\\ \label{eqn:last}
&\leq \frac{l}{2}\left( \frac{I_i}{N_0}\right)^2,
\end{align}
where $I_{i,t}$ is the power of the interference received by WN $i$ at time $t$, $I_i \coloneqq \max_t\{I_{i,t}\}$, and \eqref{eqn:last} follows from the inequality $\log(1+x) \geq x - \frac{x^2}{2}$ for $x\geq0$.

To satisfy the covertness constraint, we upper bound \eqref{eqn:last} by $\delta$. Based on the aforementioned detouring method and the generated random codebook, we next present a value for $P_{\mathrm{MH}}$ such that \eqref{eqn:last} is upper bounded by $\delta$. 
\begin{lemma}\label{lem:power}
If $\gamma = \frac{\kappa}{2}$, and 
\begin{align}\label{eqn:pMH}
P_{\mathrm{MH}} =  k\cdot n^{(\frac{1}{2}-\frac{\kappa}{2})(\alpha-2)}\sqrt{\frac{\delta}{l}}\left(\sqrt{\frac{2\log n}{n}}\right)^{\alpha},
\end{align}
where $k$ is a constant independent of $\alpha$, $\kappa$, $\delta$, $l$, and $n$, then \eqref{eqn:last} is upper bounded by $\delta$ for all $i$.
\end{lemma}

\begin{proof}
We first derive an upper bound on $I_i$. 
\begin{figure}
 \centering
  {
  \includegraphics[width=80mm]{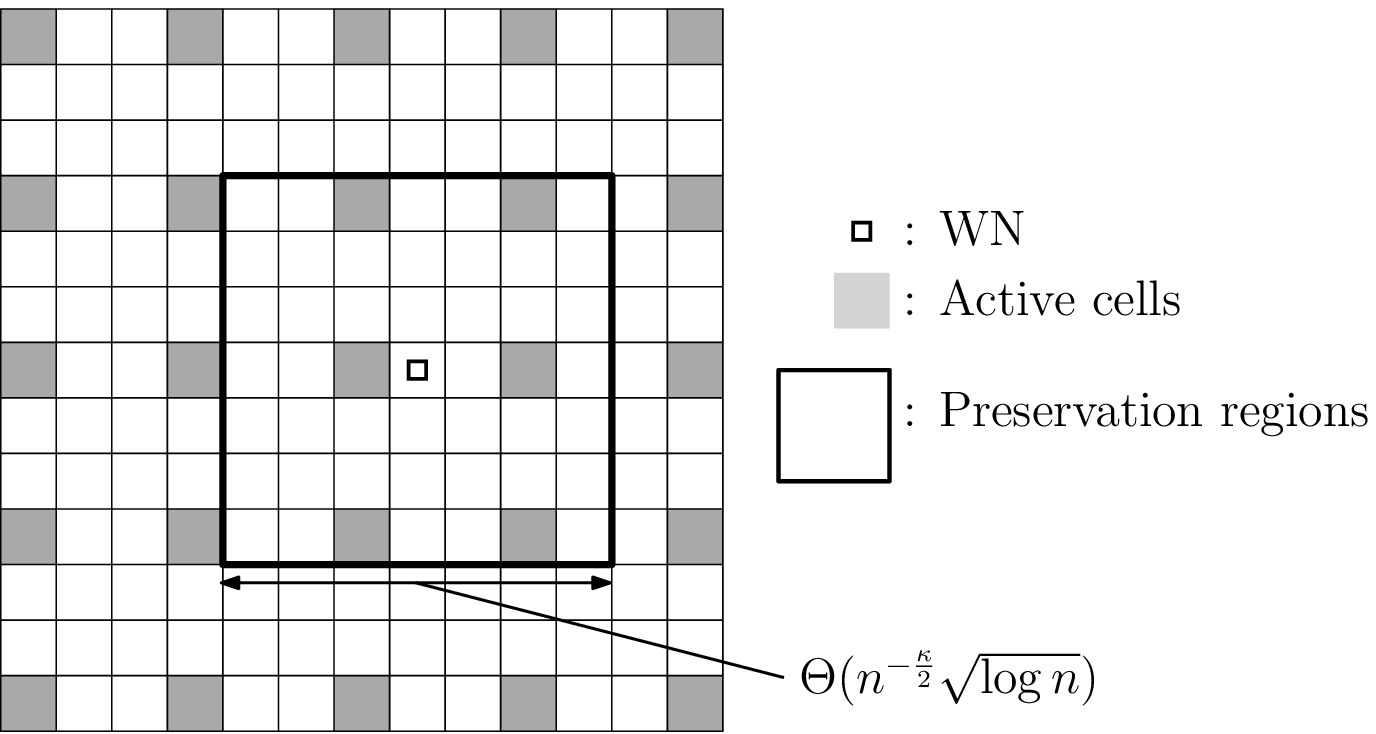}}
  \caption{Example of the worst case that the interference to a WN is strongest under the 9-TDMA scheme.} \label{fig:9TDMA}
\end{figure}
We consider the worst case scenario in which the interference to WN $i$ is the strongest. As shown in Fig.~\ref{fig:9TDMA}, this occurs when WN $i$ is adjacent to an active cell and the 9-TDMA scheme is employed. Then,
\begin{align}
 I_{i} &< GP_{\mathrm{MH}}\sum_{i=n^{\frac{1}{2}-\frac{\kappa}{2}}}^n8i\left((3i-2)\sqrt{\frac{2\log n}{n}}\right)^{-\alpha}\\ \label{eqn:sum}
&= GP_{\mathrm{MH}}\left(\sqrt{\frac{2\log n}{n}}\right)^{-\alpha}\sum_{i=n^{\frac{1}{2}-\frac{\kappa}{2}}}^n8i(3i-2)^{-\alpha}.
\end{align}
To calculate $\sum_{i=n^{\frac{1}{2}-\frac{\kappa}{2}}}^n8i(3i-2)^{-\alpha}$, we first show that this value tends to zeros as $n$ tends to infinity.
\begin{align}
\sum_{i=n^{\frac{1}{2}-\frac{\kappa}{2}}}^{n} 8i(3i-2)^{-\alpha}  \label{eqn:simplebound}
&\leq \sum_{i=n^{\frac{1}{2}-\frac{\kappa}{2}}}^{n} 8i(2i)^{-\alpha} \\
&= \sum_{i=n^{\frac{1}{2}-\frac{\kappa}{2}}}^{n} 2^{3-\alpha}\cdot i^{-\alpha+1} \\
&= 2^{3-\alpha}\sum_{i=n^{\frac{1}{2}-\frac{\kappa}{2}}}^{n} \frac{1}{i^{\alpha-1}} \\ \label{eqn:pseries}
&= 2^{3-\alpha} \left( \sum_{i=1}^{n}\frac{1}{i^{\alpha-1}} - \sum_{i=1}^{n^{\frac{1}{2}-\frac{\kappa}{2}}}\frac{1}{i^{\alpha-1}}   \right),
\end{align}
where \eqref{eqn:simplebound} holds if $n$ is sufficiently large. Recall that $\alpha > 2$. The two summations in \eqref{eqn:pseries} are $p$-series when $n$ tends to infinity and converge because $\alpha - 1 > 1$. Therefore, these two summations converge to the same value and thus the value in \eqref{eqn:pseries} converges to zero as $n$ tends to infinity. 

Instead of the exact expression of $\sum_{i=n^{\frac{1}{2}-\frac{\kappa}{2}}}^n8i(3i-2)^{-\alpha}$, we derive the order of $\sum_{i=n^{\frac{1}{2}-\frac{\kappa}{2}}}^n8i(3i-2)^{-\alpha}$. Because $8x(3x-2)^{-\alpha}$ is monotonically decreasing in $x$ for sufficiently large $x$ and $\alpha > 2$, we have
\begin{align}
\int_{n^{\frac{1}{2}-\frac{\kappa}{2}}}^{n}8x(3x-2)^{-\alpha}\, \mathrm{d}x  &< \sum_{i=n^{\frac{1}{2}-\frac{\kappa}{2}}}^n8i(3i-2)^{-\alpha} \\
&< \int_{n^{\frac{1}{2}-\frac{\kappa}{2}}}^{n}8(x-1)(3x-5)^{-\alpha}\, \mathrm{d}x.
\end{align}
We obtain the order of $\sum_{i=n^{\frac{1}{2}-\frac{\kappa}{2}}}^n8i(3i-2)^{-\alpha}$ by deriving the orders of the above two integrals. First consider the integral in the left side. When $n$ is sufficiently large,
\begin{align} 
\int_{n^{\frac{1}{2}-\frac{\kappa}{2}}}^{n}8x(4x)^{-\alpha}\, \mathrm{d}x &< \int_{n^{\frac{1}{2}-\frac{\kappa}{2}}}^{n}8x(3x-2)^{-\alpha}\, \mathrm{d}x \\
&< \int_{n^{\frac{1}{2}-\frac{\kappa}{2}}}^{n}8x(2x)^{-\alpha}\, \mathrm{d}x.
\end{align}
Thus, we have
\begin{align}
\int_{n^{\frac{1}{2}-\frac{\kappa}{2}}}^{n}8x(3x-2)^{-\alpha}\, \mathrm{d}x < c\int_{n^{\frac{1}{2}-\frac{\kappa}{2}}}^{n} x^{1-\alpha}\, \mathrm{d}x,
\end{align}
where $c$ is a constant independent of $\alpha$ and $n$. Furthermore,
\begin{align}
\int_{n^{\frac{1}{2}-\frac{\kappa}{2}}}^{n} x^{1-\alpha}\, \mathrm{d}x &= \left [ \frac{1}{2-\alpha} x^{2-\alpha} \right ]_{n^{\frac{1}{2}-\frac{\kappa}{2}}}^{n} \\
&= \frac{1}{2-\alpha} \left( n^{2-\alpha} - n^{(\frac{1}{2}-\frac{\kappa}{2})(2-\alpha)} \right) \\
&= \frac{1}{\alpha-2} \left( \frac{1}{n^{(\frac{1}{2}-\frac{\kappa}{2})(\alpha-2)}} - \frac{1}{n^{\alpha-2}} \right),
\end{align}
where $\frac{1}{n^{(\frac{1}{2}-\frac{\kappa}{2})(\alpha-2)}}$ dominates $\frac{1}{n^{\alpha-2}}$ if $n$
is sufficiently large. Therefore, we have
\begin{align} \label{eqn:leftintegral}
\int_{n^{\frac{1}{2}-\frac{\kappa}{2}}}^{n}8x(3x-2)^{-\alpha}\, \mathrm{d}x = \Theta\left(\frac{1}{n^{(\frac{1}{2}-\frac{\kappa}{2})(\alpha-2)}}\right).
\end{align}
The case of $\int_{n^{\frac{1}{2}-\frac{\kappa}{2}}}^{n}8(x-1)(3x-5)^{-\alpha}\, \mathrm{d}x$ can be done similarly, and we also have
\begin{align} \label{eqn:rightintegral}
\int_{n^{\frac{1}{2}-\frac{\kappa}{2}}}^{n}8(x-1)(3x-5)^{-\alpha}\, \mathrm{d}x = \Theta\left(\frac{1}{n^{(\frac{1}{2}-\frac{\kappa}{2})(\alpha-2)}}\right).
\end{align}
Therefore, we can conclude that
\begin{align}
\sum_{i=n^{\frac{1}{2}-\frac{\kappa}{2}}}^n8i(3i-2)^{-\alpha} = \Theta\left(\frac{1}{n^{(\frac{1}{2}-\frac{\kappa}{2})(\alpha-2)}}\right),
\end{align}
and thus
\begin{align} \label{eqn:iub}
I_{i} < \frac{cG P_{\mathrm{MH}}}{n^{(\frac{1}{2}-\frac{\kappa}{2})(\alpha-2)}}\left(\sqrt{\frac{2\log n}{n}}\right)^{-\alpha},
\end{align}
for sufficiently large $n$ and a constant $c$ independent of $\alpha$, $\kappa$, $P_{\mathrm{MH}}$, and $n$. 

To upper bound $I_i$ by $\sqrt{2}N_0\sqrt{\frac{\delta}{l}}$, we instead upper bound the right hand side in \eqref{eqn:iub} by $\sqrt{2}N_0\sqrt{\frac{\delta}{l}}$. This gives
\begin{align}
P_{\mathrm{MH}} \leq k\cdot n^{(\frac{1}{2}-\frac{\kappa}{2})(\alpha-2)}\sqrt{\frac{\delta}{l}}\left(\sqrt{\frac{2\log n}{n}}\right)^{\alpha},
\end{align}
where $k$ is a constant independent of $\alpha$, $\kappa$, $\delta$, $l$, and $n$.
\end{proof}

\subsubsection{An Achievable Throughput Scaling}\label{subsubsec:throughputMH}
Finally, we obtain an achievable throughput scaling of the detoured MH scheme. If each LN uses the transmit power $P_{\mathrm{MH}}$ as given in Lemma \ref{lem:power}, we may further lower bound the rate in a hop (continuing from Lemma \ref{lem:cellrate}) as follows: 
\begin{align}\label{eqn:rate}
R_{\mathrm{hop}} \geq \frac{1}{9} \log \left(1 + \frac{k'  \frac{n^{(\frac{1}{2}-\frac{\kappa}{2})(\alpha-2)}}{\sqrt{l}}}{N_0 + k^{''}  \frac{n^{(\frac{1}{2}-\frac{\kappa}{2})(\alpha-2)}}{\sqrt{l}}}\right) ,
\end{align}
where $k'$ and $k^{''}$ are constants independent of $\alpha$, $\kappa$, $l$, and $n$. 

When $\frac{n^{(\frac{1}{2}-\frac{\kappa}{2})(\alpha-2)}}{\sqrt{l}}$ vanishes as $n$ tends to infinity, the received power in each LN is weak compared to the noise level and hence $R_{\mathrm{hop}} = \Theta \left( \frac{n^{(\frac{1}{2}-\frac{\kappa}{2})(\alpha-2)}}{\sqrt{l}} \right)$ because $\log(1+x) = x + o(x^2)$ as $x \to 0$. For the other case, $R_{\mathrm{hop}}$ is a constant.

Then, an achievable throughput scaling of the detoured MH scheme is given in the following lemma.
\begin{lemma}\label{lem:MH}
Fix $\epsilon > 0$. The following aggregate throughput $T(n,\kappa)$ is achievable under the covertness constraint w.h.p.:
\begin{align}
T(n, \kappa) = \Omega \left(  n^{\frac{1}{2}-\epsilon} \cdot \frac{n^{(\frac{1}{2}-\frac{\kappa}{2})(\alpha-2)}}{\sqrt{l}}  \right),
\end{align}
if $\frac{n^{(\frac{1}{2}-\frac{\kappa}{2})(\alpha-2)}}{\sqrt{l}}$ vanishes as $n$ tends to infinity, and
\begin{align}
T(n, \kappa) =\Omega \left( n^{\frac{1}{2}-\epsilon} \right),
\end{align}
if $\frac{n^{(\frac{1}{2}-\frac{\kappa}{2})(\alpha-2)}}{\sqrt{l}}$ does not vanish as $n$ tends to infinity. 
\end{lemma}

\begin{proof}
The proof is straightforward from Lemma \ref{lem:detouringmethod} and \eqref{eqn:rate}.
\end{proof}

We observe that when $l$ is small relative to the number of LNs, the same throughput scaling of the MH without the covertness constraint can be achieved.

\subsection{Modified HC Scheme}\label{subsec:HC}
In this subsection, we present our modified HC scheme. We first summarize the original HC scheme proposed in \cite{Ozgur:07} and present a fundamental problem that is generated when one utilizes the original HC scheme in the covert communication scenario. We propose several modifications to the original HC scheme that can be the solution to the fundamental problem mentioned and at the same time do not degrade the asymptotic performance of the original scheme, and yet achieve the upper bounds on the throughput.
\subsubsection{Original HC Scheme}\label{subsubsec:originalHC}
There are two key ideas in the original HC scheme: Utilizing distributed MIMO systems by node cooperation and constructing the hierarchical structure of the distributed MIMO systems. We first summarize the application of the distributed MIMO system in a network unit square area as follows:
\begin{itemize}
\item Divide the network into square clusters with area of $\frac{M}{n}$ where $M$ will be optimized later.
\item \textbf{Phase 1}: In each cluster, each source node distributes its data to the other nodes in the cluster according to a certain network protocol. Clusters work in parallel.
\item \textbf{Phase 2}: The distributed data of each source node is transmitted to the cluster that contains the corresponding destination node by the distributed MIMO transmission between the two clusters that contain the source and destination nodes. These MIMO transmissions are performed serially until every source node finishes its transmission. In this phase, because each node operates for an $\frac{M}{n}$ fraction of the time, it uses bursty transmit power for the MIMO transmissions when a certain average transmit power constraint is given. 
\item \textbf{Phase 3}: In each cluster, each destination node collects the MIMO observations from the nodes in the same cluster and decodes the data. The protocol in phase 1 is used, and clusters work in parallel.
\end{itemize}
By optimizing $M$, the above procedure constitutes a new network protocol that can achieve a higher throughput scaling than that of the protocol that is used in phases 1 and 3. 

This procedure can be repeated and thus by applying this procedure hierarchically, the HC scheme is constructed. It turns out that the HC scheme requires an average transmit power of $\frac{P}{n}$ to achieve (almost) linear scaling. If the available average transmit power is less than $\frac{P}{n}$, one approach suggested in \cite{Ozgur:07} is to use a {\em bursty} HC scheme that allows bursty transmission, i.e., run the HC scheme for a fraction of time and stay idle for the remaining fraction to satisfy the average transmit power constraint of $< \frac{P}{n}$.

\subsubsection{Bursty Transmission Increases the Detectability of the Covert Communication}\label{subsubsec:bursty}
Under the covertness constraint \eqref{eqn:covertness}, the transmit powers of the LNs are constrained, and the resulting average transmit power constraint is more stringent than the average transmit power constraint of $<\frac{P}{n}$. Under this condition, the existing bursty HC scheme can be used. However, using the bursty transmit power is not a suitable approach in the covert communication setup. 

Similarly to the case of the detoured MH scheme, assume that each LN uses a codebook generated according to the distribution $\mathcal{CN}(0,P_{\mathrm{HC}})$. Then, at WN $i$, we have
\begin{align} 
D\left(\left. \hat{Q}_{Z_i^l} \right \| Q_{N_i'}^{\times l}\right) \label{eqn:first}
&= \sum_{t=1}^l D\left(\left. \hat{Q}_{Z_{i,t}} \right \| Q_{N_i'}\right) \\ \label{eqn:hclast}
&\leq \frac{1}{2}\sum_{t=1}^l\left( \frac{I_{i,t}}{N_0}\right)^2,
\end{align}
where $I_{i,t}$ is the power of the interference received by WN $i$ at time $t$. As shown in \eqref{eqn:hclast}, the upper bound on the relative entropy at each time slot in \eqref{eqn:first} is proportional to the square of the power of the interference. This shows that using the transmit power in bursty sense significantly increases the detectability in the covert communication scenario.
As mentioned above, the bursty transmission in the HC scheme is not suitable for covert network communication in adhoc networks. Thus, instead of using the bursty HC scheme, in our modified HC scheme, each LN uses the transmit power as regularly as possible over time. To realize this, we generalize \cite[Lemma~3.1]{Ozgur:07} that considers the case in which an average transmit power constraint of $<\frac{P}{n}$ is given to the case in which a more stringent average transmit power constraint is imposed. By this generalization, we can verify that by using a low transmit power, i.e., using a codebook generated according to a zero mean complex Gaussian distribution with a small variance, it is possible to achieve the same throughput scaling that can be achieved by the bursty HC scheme. Details are given in Section \ref{subsubsec:generalization}. 

Even if we do not use the bursty HC scheme, the MIMO transmissions from each cluster in phase 2 are operated in a bursty sense, and this increases the detectability. Our solution consists of two steps. We first distribute the MIMO transmisssions of each cluster evenly over phase 2. Second, since the LNs still use the bursty transmit power for MIMO transmission, we decrease the power of the MIMO transmissions up to a certain level so that the covertness constraint in \eqref{eqn:covertness} is satisfied. We observe that this does not influence the throughput scaling. The details are presented in Sections \ref{subsubsec:scheduling} and \ref{subsubsec:MIMOpower}.      

\subsubsection{Generalization of \cite[Lemma~3.1]{Ozgur:07}}\label{subsubsec:generalization}
The paper \cite{Ozgur:07} considers a network where the nodes are distributed according to a uniform distribution. However, the results in \cite{Ozgur:07} also hold in a network in which the nodes are distributed according to a p.p.p. Before generalizing \cite[Lemma~3.1]{Ozgur:07}, we restate \cite[Lemma~3.1]{Ozgur:07} for a network where $n_{\mathrm{l}}$ LNs are distributed according to a p.p.p. of density $n$. 

\begin{lemma}[\"{O}zg\"{u}r-L\'{e}v\^{e}que-Tse]\label{lem:Ozgurlemma}
Consider a network with $n_{\mathrm{l}}$ LNs and $\alpha > 2$ where the received signal of LN $i$ is given as  
\begin{align}
Y_i = \sum_{k=1}^{n_{\mathrm{l}}}H_{ik}X_k + N_i + V_i,
\end{align}
where $V_i$ is the interference signal from external sources received by LN $i$. The $V_i$ for $i = 1,2,...,n_{\mathrm{l}}$ are uncorrelated zero-mean stationary and ergodic random processes of which powers are upper bounded by a constant $V$ independent of $n$. 

Assume that there is a network scheme achieving the throughput scaling of $\Theta(n^b)$ w.h.p. for $0\leq b < 1$ under the average transmit power constraint of $<\frac{P}{n}$. Then, the increased throughput scaling of $\Theta(n^{\frac{1}{2-b}})$ is also achievable w.h.p. under the average transmit power constraint of $<\frac{P}{n}$ by constructing a new network scheme.
\end{lemma}
See \cite{Ozgur:07} for details.
By applying the above lemma iteratively, one can achieve a throughput scaling of $\Theta(n^{1-\epsilon})$ w.h.p., where $\epsilon > 0$ is an arbitrarily small constant. Consider the more stringent average transmit power constraint than the average power constraint of $<\frac{P}{n}$, i.e., power constraint of $<\frac{P}{n^{\gamma}}$ for $\gamma > 1$. Then, by using an approach in \cite{Ozgur:07} based on the bursty HC scheme, this scheme achieves a throughput scaling of $\Theta(n^{2-\gamma-\epsilon})$ where $\epsilon > 0$ is an arbitrarily small constant.

Because the bursty transmission increases the detectability of the covert communication, we do not use the bursty HC scheme but generalize \cite[Lemma~3.1]{Ozgur:07} to the case with a more stringent power constraint. 

\begin{lemma}[Generalization]\label{lem:general}
Consider a network with $n_{\mathrm{l}}$ LNs and $\alpha > 2$ where the received signal of LN $i$ is given as  
\begin{align}
Y_i = \sum_{k=1}^{n_{\mathrm{l}}}H_{ik}X_k + N_i + V_i,
\end{align}
where $V_i$ is defined similarly as in Lemma \ref{lem:Ozgurlemma}. 

Assume that there is a network scheme achieving the throughput scaling of $\Theta(n^b)$ w.h.p. for $-\infty < b < 1$ under an average transmit power constraint of $<\frac{P}{n^{\gamma}}$ for $ \gamma \geq 1$. Then, the increased throughput scaling of $\Theta(n^{\frac{\gamma}{2-b}+1-\gamma})$ is also achievable w.h.p. under the average transmit power constraint of $\frac{P}{n^{\gamma}}$ by constructing a new transmission scheme for the network. 
\end{lemma}

\begin{proof}
The proof is similar to the proof for \cite[Lemma~3.1]{Ozgur:07}. Divide the network into square clusters with area of $\frac{M}{n}$. In phases $1$ and $3$, the network scheme achieving the throughput scaling of $\Theta(n^b)$ is used. In phase 2, we multiply the power of the MIMO transmissions in \cite[Lemma~3.1]{Ozgur:07} by $n^{1-\gamma}$ and adapt the rate. Then, the linear scaling law \cite[Appendix~I]{Ozgur:07} for the MIMO gain still holds. Assume that each LN wants to send $\Theta(Mn^{1-\gamma})$ bits. Then, $\Theta(M^{2-b}n^{1-\gamma})$ time slots are needed in phases 1 and 3, and $\Theta(n)$ time slots are needed in phase 2. The resulting throughput is $\frac{c_1 \cdot Mn^{2-\gamma}}{c_2 \cdot M^{2-b}n^{1-\gamma} + c_3 \cdot n}$, where $c_1$, $c_2$, and $c_3$ are constants independent of $n$. If we set $M = n^{\frac{\gamma}{2-b}}$, the resulting throughput scaling is $\Theta(n^{\frac{\gamma}{2-b}+1-\gamma})$.    
\end{proof}

We can see that the case of $\gamma = 1$ corresponds to Lemma \ref{lem:Ozgurlemma}. If we apply this lemma iteratively, the exponents over $n$ of the achievable throughput scalings increase as $b$, $\frac{\gamma}{2-b} + 1 -\gamma$, $\frac{\gamma + \gamma(1-b)}{2-b+\gamma(1-b)}+1-\gamma$, $\frac{\gamma + \gamma(1-b) + \gamma^2(1-b)}{2-b + \gamma(1-b) + \gamma^2(1-b)} + 1 - \gamma$,$...$, $\frac{\gamma + \sum_{k}\gamma^{k}(1-b)}{2-b + \sum_{k}\gamma^{k}(1-b)} + 1 - \gamma$,$...$, and
\begin{align}
\frac{\gamma + \sum_{k}\gamma^{k}(1-b)}{2-b + \sum_{k}\gamma^{k}(1-b)} = 1 - \frac{2-b-\gamma}{2-b + \sum_{k}\gamma^{k}(1-b)}.
\end{align} 
The above term tends to one as $k$ increases and thus the throughput scaling of $\Theta(n^{2-\gamma-\epsilon})$ is achievable. We present an example of this lemma in the following.

\begin{example}
Consider a na\"ive TDMA scheme in which one LN transmits its data to its destination at one time slot using instantaneous transmit power $\frac{P}{n^{\gamma-1}}$. Then, an aggregate throughput scaling of $\Theta(n^{1-\gamma})$ is achievable. By applying Lemma \ref{lem:general} repeatedly, the exponents over $n$ of the achievable throughput scalings increase as $1-\gamma$, $\frac{\gamma}{1+\gamma} + 1 - \gamma$, $\frac{\gamma + \gamma^2}{1+\gamma + \gamma^2} + 1 - \gamma$,$...$, and these values tend to $2-\gamma$.
\end{example}

\subsubsection{Scheduling of the MIMO Transmissions}\label{subsubsec:scheduling}
\begin{figure}
 \centering
  {
  \includegraphics[width=90mm]{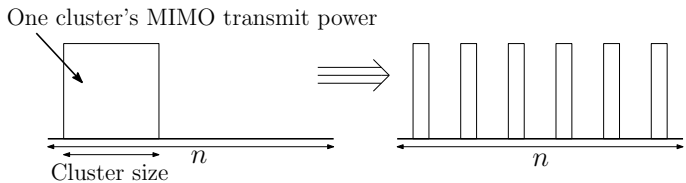}}
  \caption{Distribute the MIMO transmissions from each cluster evenly in phase 2.} \label{fig:mimo_ex}
\end{figure}
Although we do not use the bursty HC scheme, the MIMO transmissions in phase 2 are performed in a bursty sense and thus the detectability is still high. To solve this problem, we first want to avoid the scenario in which a certain WN receives the interference by the MIMO transmissions from a certain cluster for $l$ consecutive transmissions. Thus, we distribute the $M$ MIMO transmissions from each cluster as evenly as possible over $n$ time slots i.e., in every window of $l$ time slots, approximately $\frac{lM}{n}$ MIMO transmissions from each cluster are as shown in Fig. \ref{fig:mimo_ex}.

\subsubsection{Decreasing the Power of the MIMO Transmissions does not Degrade the Throughput Scaling}\label{subsubsec:MIMOpower} 
We can cancel out the increase of the detectability due to the bursty MIMO transmission by decreasing the power of the MIMO transmissions, i.e., multiply the original power for MIMO transmission by $\sqrt{\frac{M}{n}}$ where $M$ is the size of a cluster. One expects that this degrades the performance of the network scheme. However, we now show that this does not degrade the achievable throughput scaling.
\begin{lemma}\label{lem:scalingdown}
In phase 2 of the HC scheme, assume that each cluster contains approximately $M$ LNs. Then, multiplying the power of the MIMO transmissions by $\sqrt{\frac{M}{n}}$ does not degrade the achievable throughput scaling. 
\end{lemma}

Intuitively, as the number of times we invoke Lemma \ref{lem:general} increases, the optimal value of $M$ tends to $n$. This means that the effect of reducing the power diminishes and thus the performance of the MIMO transmissions tends to the performance without the reduction of the power. For the proof, assume that the average transmit power constraint of $<\frac{P}{n}$ is given for convenience. We then show that near linear scaling of the throughput is achievable even if we multiply the power of the MIMO transmissions in Lemma \ref{lem:Ozgurlemma} by $\sqrt{\frac{M}{n}}$. Assume that each LN wants to send $\Theta(M\sqrt{\frac{M}{n}})$ bits. Then, in phases 1 and 3, $\Theta(M^{2-b}\sqrt{\frac{M}{n}})$ time slots are needed, and in phase 2, $\Theta(n)$ time slots are needed because of the linear scaling law for the MIMO gain. Then, the resulting throughput is $\frac{c_1 \cdot nM \sqrt{M/n}}{c_2 \cdot M^{2-b}\sqrt{M/n} + c_3 \cdot n}$, where $c_1$, $c_2$, and $c_3$ are constants independent of $n$. If we set $M = n^{\frac{3}{5-2b}}$, the resulting achievable throughput scaling is $\Theta(n^{\frac{9}{10-4b}-\frac{1}{2}})$. 
\begin{figure}
 \centering
  {
  \includegraphics[width=48mm]{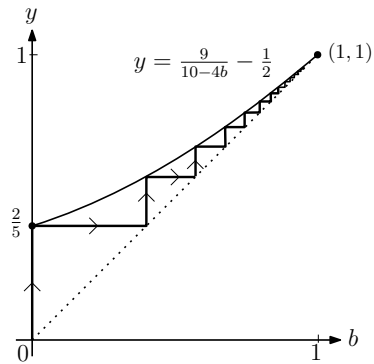}}
  \caption{Example of the worst case that the interference to a WN is strongest under the 9-TDMA scheme.} \label{fig:linear}
\end{figure}
By applying this method iteratively, the throughput scaling of $\Theta(n^{1-\epsilon})$ is achievable as shown in Fig. \ref{fig:linear}.

\subsubsection{Modified HC in Covert Communication}\label{subsubsec:modifiedHC}
Now consider our setting. In the presence of the covertness constraint \eqref{eqn:covertness} and the preservation regions, the achievable throughput scaling by the modified HC scheme is given in the following lemma. 
\begin{lemma} \label{lem:modifiedHC}
Fix $\epsilon > 0$. The following aggregate throughput $T(n,\kappa)$ is achievable under the covertness constraint w.h.p.:
\begin{align}
T(n,\kappa) = \Omega \left(n^{2-\frac{\alpha}{2}-\epsilon} \cdot \frac{n^{(\frac{1}{2}-\frac{\kappa}{2})(\alpha-2)}}{\sqrt{l}} \right).
\end{align}
\end{lemma}

\begin{proof}
For convenience of the proof, similar to \eqref{eqn:pMH}, assume that the average transmit power over an arbitrary window of $l$ channel uses at each LN is upper bounded by
\begin{align}
\frac{n^{(\frac{1}{2}-\frac{\kappa}{2})(\alpha-2)}}{2 \log n}\sqrt{\frac{\delta}{l}}\left(\sqrt{\frac{2\log n}{n}}\right)^{\alpha}.
\end{align}  
The above value is smaller than $\frac{P}{n}$ regardless of $\alpha$, $\kappa$, and $l$. Then, using the average transmit power above, the bursty HC scheme without considering the covertness constraint can achieve the throughput scaling of $\Theta \left( n^{2-\frac{\alpha}{2}-\epsilon}\cdot \frac{n^{(\frac{1}{2}-\frac{\kappa}{2})(\alpha-2)}}{\sqrt{l}}\right)$. 
In our modified HC scheme, the modifications in Sections \ref{subsubsec:generalization} to \ref{subsubsec:MIMOpower} are applied, and we observed that these modifications do not degrade the performance of the HC scheme compared to the bursty HC scheme. Furthermore, according to Lemma \ref{lem:power} and \eqref{eqn:hclast}, we can verify that the covertness constraint is satisfied. 

Now we have to consider the preservation regions. 
\begin{figure}
 \centering
  {
  \includegraphics[width=40mm]{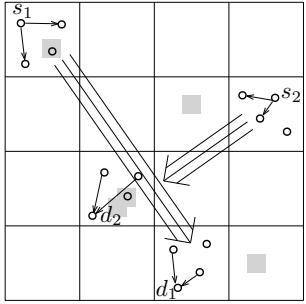}}
  \caption{The modified HC scheme in the presence of the preservation regions. The gray squares are the preservation regions.} \label{fig:hie}
\end{figure}
Let each of the LN pairs that the LSN or the LDN is in the preservation regions be in outage. In the HC scheme, there are two cases of the (sub)clusters for MIMO transmission: The clusters with size smaller than the size of a preservation region and the clusters with size larger than the size of a preservation region. If the size of the MIMO cluster is smaller than the size of a preservation region, let the region of the clusters that partially or completely belong to the preservation regions (and thus are impaired) be extra preservation regions. This does not influence the scaling of the throughput. If the size of the cluster is larger than the size of a preservation region, the LNs outside the preservation regions supplement the LNs in the preservation regions by increasing the rate as shown in Fig. \ref{fig:hie}. Because the fraction of the LNs in the preservation regions in a cluster is negligible, this does not influence the scaling of the throughput. Therefore, by adding the negligible fraction of extra preservation regions and supplementing the LNs in the preservation regions, the modified HC scheme can operate without any degradation of the throughput scaling.       
\end{proof}

\subsection{Hybrid HC-MH}\label{subsec:hybrid}
In this subsection, we present a modified version of the hybrid HC-MH scheme proposed in \cite{regime}. We first summarize the original hybrid HC-MH scheme, and then present a modified scheme in the presence of the covertness constraint and the preservation regions.

\subsubsection{Original Hybrid HC-MH Scheme}\label{subsubsec:originalhybrid}
\begin{figure}
 \centering
  {
  \includegraphics[width=70mm]{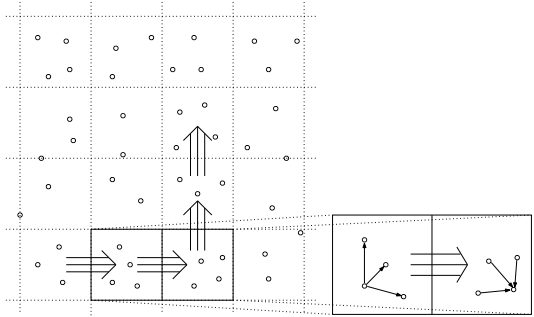}}
  \caption{Example of the hybrid HC-MH scheme. The HC scheme is used in each square cell and the MH scheme is used for the global MH MIMO transmissions.} \label{fig:originalhybrid}
\end{figure}
Consider a network in which the long-range $\SNR$ is insufficient and thus the HC scheme is outperformed by the MH scheme. In this network, if the short-range $\SNR$ (the received $\SNR$ in the MH scheme) is non-vanishing, the MH scheme cannot efficiently utilize the sufficient received $\SNR$. One approach to efficiently utilize the $\SNR$ is constructing small size of the HC structures locally and delivering the data by the MH MIMO transmissions as shown in Fig. \ref{fig:originalhybrid}. The hybrid scheme in a network of unit square area is summarized as follows: 
\begin{itemize}  
\item Divide the network into square cells with area of $\frac{M}{n}$ where $M$ will be optimized later. 
\item Locally in each cell, the HC scheme is performed to distribute the data of each node (or the data to be relayed) and gather the data to be decoded.
\item Globally, the data of each node that is distributed to the nodes in the cell is delivered to the neighboring cell by distributed MIMO transmission. The sequential MIMO transmissions are performed based on the MH scheme as described in Section \ref{subsubsec:originalMH}.  
\end{itemize}

Assume that the average transmit power is upper bounded by $P$, and let $P_{\mathrm{short}} \coloneqq Pn^{\alpha/2}$, i.e., short-range $\SNR$. Then, the optimal value of $M$ is $P_{\mathrm{short}}^{1/(\alpha/2-1)}$, and the resulting achievable throughput scaling is $\Theta(n^{\frac{1}{2}-\epsilon} \cdot P_{\mathrm{short}}^{\frac{1}{\alpha-2}})$.

\subsubsection{Detoured Hybrid HC-MH Scheme}\label{subsubsec:detouredhybrid}
Now consider our setting and assume that $\frac{n^{(\frac{1}{2}-\frac{\kappa}{2})(\alpha-2)}}{\sqrt{l}}$ is non-vanishing. Our modified hybrid scheme, termed {\em detoured} hybrid scheme is based on the original hybrid scheme and achieves the throughput scaling that is presented in the following lemma. 
\begin{lemma} \label{lem:originalhybrid}
Fix $\epsilon > 0$. If $\frac{n^{(\frac{1}{2}-\frac{\kappa}{2})(\alpha-2)}}{\sqrt{l}}$ does not vanish as $n$ tends to infinity, the following aggregate throughput $T(n,\kappa)$ is achievable under the covertness constraint w.h.p.:  
\begin{align}
T(n,\kappa) = \Omega \left(n^{\frac{1}{2}-\epsilon} \left( \frac{n^{(\frac{1}{2}-\frac{\kappa}{2})(\alpha-2)}}{\sqrt{l}}\right)^{\frac{1}{\alpha-2}} \right). 
\end{align}
\end{lemma} 
\begin{proof}
For convenience, similar to the proof for Lemma \ref{lem:modifiedHC}, assume that the average transmit power over an arbitrary window of $l$ channel uses at each LN is upper bounded by $\frac{n^{(\frac{1}{2}-\frac{\kappa}{2})(\alpha-2)}}{2 \log n}\sqrt{\frac{\delta}{l}}\left(\sqrt{\frac{2\log n}{n}}\right)^{\alpha}$. Then, the role of $\frac{n^{(\frac{1}{2}-\frac{\kappa}{2})(\alpha-2)}}{\sqrt{l}}$ is similar to the role of $P_{\mathrm{short}}$. Set $M$ tbe $\left(  \frac{n^{(\frac{1}{2}-\frac{\kappa}{2})(\alpha-2)}}{\sqrt{l}} \right)^{\frac{1}{\alpha/2-1}}$ and the preservation region exponent $\gamma$ to be $\frac{\kappa}{2}$. Then, any cell is smaller than its corresponding preservation region. 
\begin{figure}
 \centering
  {
  \includegraphics[width=80mm]{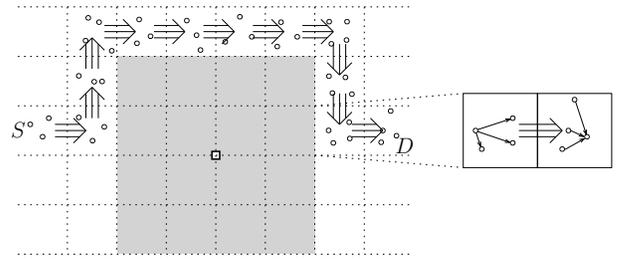}}
  \caption{Example of the modified hybrid scheme. In the global MH, the detouring method used in the modified MH scheme is applied.} \label{fig:detouredhybrid}
\end{figure}
In each cell, the modified HC scheme in Section \ref{subsubsec:modifiedHC} (without considering the preservation regions) is used, and for the global MH MIMO transmissions, the same detouring method used in the detoured MH scheme is utilized as shown in Fig. \ref{fig:detouredhybrid}. Then, the effect of the extra relaying burden is negligible and the outage fraction vanishes. Furthermore, from Lemma \ref{lem:power} and similarly to the case of the modified HC scheme, the covertness constraint is satisfied and the proof of the lemma is completed. 
\end{proof}

\section{Converse}\label{sec:conv}
In this section, we present the proofs for Theorems \ref{thm:conv1} and \ref{thm:conv2}. A cutset bound technique similar to the technique in \cite{regime} is used for the proofs. 

\subsection{A Necessary Condition on the Covertness Constraint}\label{subsec:necessary}
In this subsection, we derive a condition that is implied by the covertness constraint in \eqref{eqn:covertness}. Since the distribution of a complex random variable is equivalent to the joint distribution of the magnitude and the phase of the random variable, the relative entropy in \eqref{eqn:covertness} can be written as  
\begin{align}
D\left(\left. \hat{Q}_{Z_i^l} \right \| Q_{N_i'}^{\times l}\right) = D\left( \left. \hat{Q}_{\left(|Z_i|,e^{j\angle Z_i}\right)^l} \right \| Q_{\left(|N_i'|,e^{j\angle N_i'}\right)}^{\times l}\right). \label{eqn:qq}
\end{align}
Since the transmitters do not know the channel phases, the magnitudes and phases in \eqref{eqn:qq}  are independent. Thus, we have
\begin{align}
&D\left( \left. \hat{Q}_{\left (|Z_i|,e^{j\angle Z_i}\right )^l} \right \| Q_{\left (|N_i'|,e^{j\angle N_i'}\right )}^{\times l}\right) \nonumber \\
&= D\left(\left. \hat{Q}_{\left (|Z_i|^l,\left(e^{j\angle Z_i}\right )^l \right)} \right \| Q_{\left (|N_i'|,e^{j\angle N_i'}\right )}^{\times l}\right) \\
&= D\left (\left. \hat{Q}_{|Z_i|^l} \right \| Q_{|N_i'|}^{\times l}\right ) + D\left (\left. \hat{Q}_{(e^{j\angle Z_i})^l} \right \| Q_{e^{j\angle N_i'}}^{\times l}\right) \\ \label{eqn:phaseD}
&= D\left (\left. \hat{Q}_{|Z_i|^l} \right \| Q_{|N_i'|}^{\times l}\right ),
\end{align}
where \eqref{eqn:phaseD} is because both $\hat{Q}_{(e^{j\angle Z_i})^l}$ and $Q_{e^{j\angle N_i'}}^{\times l}$ are i.i.d. process according to the uniform distribution over $[0,2\pi]$.

Next, let $\hat{Q}_{|\bar{Z}_i|}$ denote the average distribution of $|Z_i|$ over $l$ channel uses, and let $|Z_i|_{k}$ denote the magnitude of the received signal by WN $i$ at time $k$. Then, we have
\begin{align}
&D\left (\left. \hat{Q}_{|Z_i|^l} \right \| Q_{|N_i'|}^{\times l}\right) \nonumber \\
&= -h(|Z_i|^l) + \mathbb{E}_{|Z_i|^l}\left[ \log \frac{1}{Q_{|N_i'|}^{\times l}(|Z_i|^l)  }\right] \\
&= - \sum_{k=1}^l h(|Z_i|_k | |Z_i|^{k-1}) + \mathbb{E}_{|Z_i|^l}\left[ \log \frac{1}{Q_{|N_i'|}(|Z_i|_k)  }\right] \\
&= - \sum_{k=1}^l h(|Z_i|_k | |Z_i|^{k-1}) + \mathbb{E}_{|Z_i|_k}\left[ \log \frac{1}{Q_{|N_i'|}(|Z_i|_k)  }\right] \\
&\geq  - \sum_{k=1}^l h(|Z_i|_k) + \mathbb{E}_{|Z_i|_k}\left[ \log \frac{1}{Q_{|N_i'|}(|Z_i|_k)  }\right] \\
&= \sum_{k=1}^l D\left (\left. \hat{Q}_{|Z_i|_k} \right \| Q_{|N_i'|}\right) \\ \label{eqn:Dconvexity}
&\geq  l\cdot D\left (\left. \hat{Q}_{|\bar{Z}_i|} \right \| Q_{|N_i'|}\right ), 
\end{align}
where \eqref{eqn:Dconvexity} is due to the convexity of the relative entropy. In addition, from some manipulations, we obtain that
\begin{align} \label{eqn:square}
D\left (\left. \hat{Q}_{|\bar{Z}_i|} \right \| Q_{|N_i'|}\right ) = D\left (\left. \hat{Q}_{|\bar{Z}_i|^2} \right \| Q_{|N_i'|^2}\right).
\end{align}
The proof of the above equation is in Appendix \ref{app:square}. Then, the following holds if the covertness constraint \eqref{eqn:covertness} is satisfied:
\begin{align}
D\left (\left. \hat{Q}_{|\bar{Z}_i|^2} \right \| Q_{|N_i'|^2}\right) \leq \frac{\delta}{l}. \label{eqn:single}
\end{align}

 Since the interference and the noise are independent at each WN, the first moment of the $|\bar{Z}_i|^2$ is given by $\rho_i + N_0$, where $\rho_i = \frac{1}{l}(\rho_{i1} + \rho_{i2} + \cdots + \rho_{il})$, and $\rho_{ik}$ is the total received power in WN $i$ at the $k^{\mathrm{th}}$ time slot used for the hypothesis test. Furthermore, since $|N_i'|^2$ follows an exponential distribution,
\begin{align}
& D\left (\left. \hat{Q}_{|\bar{Z}_i|^2} \right \| Q_{|N_i'|^2}\right ) \nonumber \\
&= -h(|\bar{Z}_i|^2) + \mathbb{E}_{|\bar{Z}_i|^2}\left[ \log\frac{1}{Q_{|N_i'|^2}(|\bar{Z}_i|^2)}  \right]  \\
&= -h(|\bar{Z}_i|^2) + \mathbb{E}_{|\bar{Z}_i|^2}\left[ \log\left(N_0e^{\frac{|\bar{Z}_i|^2}{N_0}} \right)  \right]\\
&= -h(|\bar{Z}_i|^2) + \log N_0 + \mathbb{E}_{|\bar{Z}_i|^2}\left[ \frac{|\bar{Z}_i|^2}{N_0}  \right]\\
&= -h(|\bar{Z}_i|^2) + \log N_0 + \frac{\rho_i + N_0}{N_0}\\ \label{eqn:exponential}
&\geq - 1 - \log (\rho_i + N_0) + \log N_0 + \frac{\rho_i + N_0}{N_0}\\ 
&= - \log \frac{\rho_i + N_0}{N_0} + \frac{\rho_i}{N_0},
\end{align}
where \eqref{eqn:exponential} is because exponential distribution maximizes the differential entropy when the first moment is given. 

To satisfy the covertness constraint \eqref{eqn:covertness}, $D\left (\left. \hat{Q}_{|\bar{Z}_i|^2} \right \| Q_{|N_i'|^2}\right )$ and therefore $\rho_i$ must tend to zero as $l$ tends to infinity. Then, 
\begin{align}
D\left (\left. \hat{Q}_{|\bar{Z}_i|^2} \right \| Q_{|N_i'|^2}\right ) \geq \frac{\rho_i^2}{2N_0^2} + o(\rho_i^2).
\end{align}
In addition, from \eqref{eqn:single}, the above yields
\begin{align}\label{eqn:constraint}
\rho_i \leq \sqrt{2}N_0\sqrt{\frac{\delta}{l}} + o(l^{-1/2}).
\end{align}

Thus, the covertness constraint implies that the $\INR$ at WN $i$ over an arbitrary window of  $l$ channel uses denoted as $\INR_{\mathrm{w},i}$, is upper bounded as follows:
\begin{align}\label{eqn:constraint}
\INR_{\mathrm{w},i} \leq \sqrt{2}\sqrt{\frac{\delta}{l}} + o(l^{-1/2}), \quad i=1,2,\ldots,n_{\mathrm{w}}.
\end{align}  

\begin{remark}
Since CSI is available only on the receivers, the transmitted signals are independent of the phase of the channels and hence the received signals are pairwise uncorrelated. Thus, $\rho_{ik}$ is the sum of the powers of the received signals at WN $i$ at the $k^{\mathrm{th}}$ time slot over a window.
\end{remark}

\subsection{Cutset Bound}\label{subsec:cutset}
Now we derive an upper bound on the throughput based on the cutset bound. At each LN, let the average transmit power over an arbitrary window of $l$ channel uses be $P_{\mathrm{CB}}$. Under condition \eqref{eqn:constraint}, the best way to allocate vanishing outage fraction to maximize $P_{\mathrm{CB}}$ is setting the preservation region around each WN to be of the same size. This is a simple $\max \min$ problem because a LN near to a small preservation region cannot use high transmit power. We set $\gamma = \frac{\kappa}{2}$ for brevity. Then, $P_{\mathrm{CB}}$ is upper bounded as 
\begin{align}\label{eqn:uppower}
P_{\mathrm{CB}} \leq c\cdot \frac{n^{(\frac{1}{2}-\frac{\kappa}{2})(\alpha-2)}}{\log n}\sqrt{\frac{\delta}{l}}\left(\sqrt{\frac{2\log n}{n}}\right)^{\alpha},
\end{align}
where $c$ is a constant independent of $\alpha$, $\kappa$, $\delta$, $l$, and $n$.

The proof of \eqref{eqn:uppower} is similar to the proof for Lemma \ref{lem:power}, but is modified appropriately for the converse step. Consider the case in which every active LN uses the average transmit power $P_{\mathrm{CB}}$ over an arbitrary window of $l$ channel uses. Then, at WN $i$, we have   
\begin{align}
 \INR_{\mathrm{w},i} &\geq \frac{GP_{\mathrm{CB}}}{N_0}\sum_{i=n^{\frac{1}{2}-\frac{\kappa}{2}}}^{\sqrt{n/\log n}}16i \log n\left( i\sqrt{\frac{2\log n}{n}}\right)^{-\alpha}\\ \label{eqn:sum}
&= kP_{\mathrm{CB}}\log n \left(\sqrt{\frac{2\log n}{n}}\right)^{-\alpha}\sum_{i=n^{\frac{1}{2}-\frac{\kappa}{2}}}^{\sqrt{n/\log n}}i^{1-\alpha} \\
&= k'P_{\mathrm{CB}}\frac{\log n}{n^{(\frac{1}{2}-\frac{\kappa}{2})(\alpha-2)}}\left(\sqrt{\frac{2\log n}{n}}\right)^{-\alpha},
\end{align}
where $k$ and $k'$ are constants independent of $\alpha$, $\kappa$, $\delta$, $l$, and $n$.
Then, from condition \eqref{eqn:constraint}, we can obtain \eqref{eqn:uppower}. Henceforth, we fix $P_{\mathrm{CB}} = c\cdot \frac{n^{(\frac{1}{2}-\frac{\kappa}{2})(\alpha-2)}}{\log n}\sqrt{\frac{\delta}{l}}\left(\sqrt{\frac{2\log n}{n}}\right)^{\alpha}$.

To derive the promised cutset bound, consider a vertical cut dividing the network into two equal halves. Let us define $T_{\mathrm{cut}}$ as the throughput of the information flow from the left half of the network to the right half of the network. 
\begin{figure}
 \centering
  {
  \includegraphics[width=50mm]{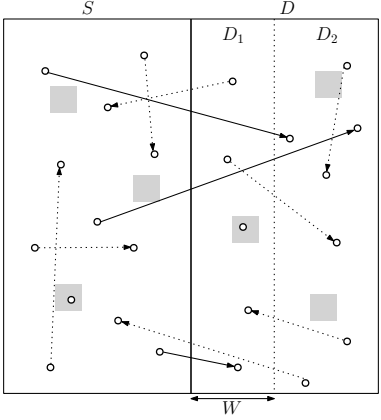}}
  \caption{Example of the information flow in the network. The arrows with bold lines represent the information delivered from the left half to the right half through the cut. The gray squares are the preservation regions, and the circles are the LNs.} \label{fig:cutset}
\end{figure}
Because the LNs are distributed and are chosen in pairs randomly, $T_{\mathrm{cut}}$ belongs to $((1-\epsilon)T(n,\kappa)/4, (1+\epsilon)T(n,\kappa)/4)$ w.h.p., where $\epsilon > 0$ is an arbitrarily small constant (Fig. \ref{fig:cutset}). 

Let $S$ be the set of LNs in the left half of the network and $D$ be the set of LNs in the right half. Then, $T_{\mathrm{cut}}$ is upper bounded by the capacity of the MIMO channel between $S$ and $D$ as follows:  
\begin{align} \label{eqn:MIMOcap}
&T_{\mathrm{cut}} \leq \nonumber \\
& \max_{ \substack{ \Sigma(H) \geq 0 \\ \mathbb{E}[\Sigma_{kk}(H)]\leq P_{\mathrm{CB}}, \forall k \in S } } \mathbb{E} \left[ \log \det \left(I + \frac{1}{N_0B}H\Sigma(H)H^{\mathrm{H}} \right)   \right],
\end{align}
where $H$ is the channel matrix whose components are 
\begin{align}
H_{ik} = \frac{\sqrt{G}}{d_{ik}^{\alpha/2}}\exp\left(j\theta_{ik} \right),\quad k \in S,\ i \in D,
\end{align}
and $A^{\mathrm{H}}$ is the Hermitian of the matrix $A$.
In \eqref{eqn:MIMOcap}$, \Sigma(\cdot)$ is a function that maps a given channel matrix to a positive semi-definite transmit covariance matrix, and $\Sigma_{kk}(H)$ is the $k^{\mathrm{th}}$ component along the diagonal of $\Sigma(H)$. 

If we define $\tilde{H}$ as the matrix whose components are
\begin{align}
\tilde{H}_{ik} = \frac{1}{(\sqrt{n})^{\alpha/2}}H_{ik},\quad k \in S,\ i \in D,
\end{align}
and define $P_{\mathrm{CB}}'$ as $\frac{n^{\alpha/2}}{N_0B} P_{\mathrm{CB}}$, then the following holds.
\begin{align}
&\max_{ \substack{ \Sigma(H) \geq 0 \\ \mathbb{E}[\Sigma_{kk}(H)]\leq P_{\mathrm{CB}}, \forall k \in S } } \mathbb{E} \left[ \log \det \left(I + \frac{1}{N_0B}H\Sigma(H)H^{\mathrm{H}} \right)   \right] \nonumber \\
&= \max_{ \substack{ \Sigma(\tilde{H}) \geq 0 \\ \mathbb{E}[\Sigma_{kk}(\tilde{H})]\leq P_{\mathrm{CB}}', \forall k \in S } } \mathbb{E} \left[ \log \det \left(I + \tilde{H}\Sigma(\tilde{H})\tilde{H}^{\mathrm{H}} \right)   \right] \\
&= \max_{ \substack{ \Sigma(\tilde{H}) \geq 0 \\ \mathbb{E}[\Sigma_{kk}(\tilde{H})]\leq 1, \forall k \in S } } \mathbb{E} \left[ \log \det \left(I + P_{\mathrm{CB}}'\tilde{H}\Sigma(\tilde{H})\tilde{H}^{\mathrm{H}} \right)   \right].
\end{align}
In the chain of equations above, $\tilde{H}$ can be interpreted as the channel matrix of the network of a square of area $n$ that is obtained by scaling our network. In addition, the role of $P_{\mathrm{CB}}' = \Theta \left( \frac{n^{(\frac{1}{2}-\frac{\kappa}{2})(\alpha-2)}}{\sqrt{l}} \right)$ is similar to that of the important factor $\SNR_s$ in \cite{regime}, and is a key parameter in the analyses contained in this paper.

Now divide the right half of the network into two rectangular regions that the size of the rectangular adjacent to the cut is $W \times 1$ as shown in Fig. \ref{fig:cutset}, where $W$ will be chosen later. Define $D_1$ as the set of LNs in this rectangular and define $D_2$ as $D \backslash D_1$. Then, due to the generalized Hadamard inequality,  
\begin{align}
&\log \det \left(I + P_{\mathrm{CB}}'\tilde{H}\Sigma(\tilde{H})\tilde{H}^{\mathrm{H}} \right)  \nonumber \\
&\leq \log \det \left(I + P_{\mathrm{CB}}'\tilde{H}_1\Sigma(\tilde{H}_1)\tilde{H}_1^{\mathrm{H}} \right) \nonumber \\
&\quad + \log \det \left(I + P_{\mathrm{CB}}'\tilde{H}_2\Sigma(\tilde{H}_2)\tilde{H}_2^{\mathrm{H}} \right),
\end{align}
where $\tilde{H}_1$ is the $|D_1|$ by $|S|$ submatrix of $\tilde{H}$ whose components are $\tilde{H}_{ik}, k \in S,\ i \in D_1$, and $\tilde{H}_2$ is the $|D_2|$ by $|S|$ submatrix of $\tilde{H}$ whose components are $\tilde{H}_{ik}, k \in S,\ i \in D_2$. Then, we have 
\begin{align}
&T_{\mathrm{cut}} \leq \nonumber \\
& \max_{ \substack{ \Sigma(\tilde{H}_1) \geq 0 \\ \mathbb{E}[\Sigma_{kk}(\tilde{H}_1)]\leq 1, \forall k \in S } } \mathbb{E} \left[ \log \det \left(I + P_{\mathrm{CB}}'\tilde{H}_1\Sigma(\tilde{H}_1)\tilde{H}_1^{\mathrm{H}} \right)   \right] \nonumber \\ \label{eqn:MMMM}
& \quad + \max_{ \substack{ \Sigma(\tilde{H}_2) \geq 0 \\ \mathbb{E}[\Sigma_{kk}(\tilde{H}_2)]\leq 1, \forall k \in S } } \mathbb{E} \left[ \log \det \left(I + P_{\mathrm{CB}}'\tilde{H}_2\Sigma(\tilde{H}_2)\tilde{H}_2^{\mathrm{H}} \right)   \right].
\end{align}

The first term in \eqref{eqn:MMMM} can be upper bounded by the sum of the multiple-input single-output (MISO) capacities as follows
 \begin{align}
& \max_{ \substack{ \Sigma(\tilde{H}_1) \geq 0 \\ \mathbb{E}[\Sigma_{kk}(\tilde{H}_1)] \leq 1, \forall k \in S } } \mathbb{E} \left[ \log \det \left(I + P_{\mathrm{CB}}'\tilde{H}_1\Sigma(\tilde{H}_1)\tilde{H}_1^{\mathrm{H}} \right)   \right] \nonumber \\  \label{eqn:MISOs1}
& \leq \sum_{i \in D_1} \log \left( 1 + \frac{nP_{\mathrm{CB}}'}{2}\sum_{k \in S}|\tilde{H}_{ik}|^2  \right).
\end{align}
In \eqref{eqn:MISOs1}, the value of $|\tilde{H}_{ik}|^2 $ is large if LN $k$ is close to the cut and thus this can result in a loose upper bound. Therefore, similarly to \cite{regime}, we assume that there is no LN in the rectangular region with size $\frac{1}{\sqrt{n}} \times 1$ to the right of the cut. The validity of this assumption can be checked in \cite[Appendix I]{regime}. Then, the right hand side in \eqref{eqn:MISOs1} can be upper bounded as 
\begin{align}
& \sum_{i \in D_1} \log \left( 1 + \frac{nP_{\mathrm{CB}}'}{2}\sum_{k \in S}|\tilde{H}_{ik}|^2  \right)  \nonumber \\ 
&= \sum_{i \in D_1} \log \left( 1 + \frac{nP_{\mathrm{CB}}'}{2n^{\alpha/2}}\sum_{k \in S}|H_{ik}|^2  \right)  \\ \label{eqn:MISOs2}
& \leq \left( W - \frac{1}{\sqrt{n}} \right) n (\log n) \left(\log \left( 1 + c_1 P_{\mathrm{CB}}' n^{1+\alpha(1/2 + \delta)} \right) \right)
\end{align} 
w.h.p. for any $\delta > 0$ where $c_1$ is a constant independent of $n$. The term $\left( W - \frac{1}{\sqrt{n}} \right)$ represents the area of the rectangular region containing $D_1$ except the rectangular region with size $\frac{1}{\sqrt{n}} \times 1$ to the right of the cut. The term $n^{\alpha(1/2 + \delta)}$ is from the fact that the distance between any two LNs in our network is larger than $\frac{1}{n^{1+\delta}}$ w.h.p. for any $\delta > 0$.  

Now consider the second term in \eqref{eqn:MMMM}. Since $\tilde{H}_2\Sigma(\tilde{H}_2)\tilde{H}_2^{\mathrm{H}}$ is a positive semi-definite matrix, $\log \lambda_i \leq \lambda_i - 1$ holds for each positive eigenvalue $\lambda_i$ of ($I+\tilde{H}_2\Sigma(\tilde{H}_2)\tilde{H}_2^{\mathrm{H}}$). Then, we have
\begin{align} \label{eqn:trace}
\log \det \left(I + P_{\mathrm{CB}}'\tilde{H}_2\Sigma(\tilde{H}_2)\tilde{H}_2^{\mathrm{H}} \right) \leq \trace \left( P_{\mathrm{CB}}'\tilde{H}_2\Sigma(\tilde{H}_2)\tilde{H}_2^{\mathrm{H}} \right).
\end{align}  
Next define $P_{\mathrm{cut}}$ as the total received power of the LNs in $D_2$ when each LN in $S$ transmits its signal with average transmit power $P_{\mathrm{CB}}$. Then, for any $\epsilon > 0$, 
\begin{align} \label{eqn:Pcut}
\trace \left( P_{\mathrm{CB}}'\tilde{H}_2\Sigma(\tilde{H}_2)\tilde{H}_2^{\mathrm{H}} \right) \leq n^{\epsilon}P_{\mathrm{cut}}.
\end{align}  
The proof of the above inequality can be checked in \cite[Lemma 5.2]{Ozgur:07}. Furthermore, if $W \neq \frac{1}{2}$, $P_{\mathrm{cut}}$ is upper bounded as 
\begin{align} \label{eqn:Pcut}
P_{\mathrm{cut}} \leq 
\begin{cases}
kP_{\mathrm{CB}}'n^{2-\alpha/2}(\log n)^2, & 2 < \alpha < 3\\
kP_{\mathrm{CB}}'\sqrt{n}(\log n)^3, & \alpha = 3 \\
kP_{\mathrm{CB}}'W^{3-\alpha}n(\log n)^2, & \alpha > 3,
\end{cases}
\end{align}
where $k$ is a constant independent of $\alpha$, $P_{\mathrm{CB}}'$, and $n$. The proof of the above inequality can be checked in the proof of \cite[Eq.~(19)]{regime}. 

Finally, by combining \eqref{eqn:MISOs2} and \eqref{eqn:Pcut}, we can obtain
\begin{align}
T_{\mathrm{cut}} &\leq \left( W - \frac{1}{\sqrt{n}} \right) n (\log n) \left(\log \left( 1 + c_1 P_{\mathrm{CB}}' n^{1+\alpha(1/2 + \delta)} \right) \right)  \nonumber \\ \label{eqn:CB} 
& \quad + n^{\epsilon}P_{\mathrm{cut}}.
\end{align}
To obtain a tight upper bound from \eqref{eqn:CB}, we can minimize the right hand side in \eqref{eqn:CB} by setting $W$ to be an appropriate function of $P_{\mathrm{CB}}$, i.e., set $W$ to be a function of $\alpha$, $\kappa$, and $l$. This is a simple $\min \max$ problem and can be solved by making the two terms asymptotically equal. To do so, simply, if $P_{\mathrm{CB}}' < 1$, we set $W = \frac{1}{\sqrt{n}}$ and $D_1 = \emptyset$. If $P_{\mathrm{CB}}' \geq 1$, we set $W = P_{\mathrm{CB}}'^{\frac{1}{\alpha-2}}$. From this, we complete the proof.

\section{Conclusion and Future Work}\label{sec:conclusion}
In this paper, we established throughput scaling laws of covert communication over wireless adhoc networks where a number of LNs and WNs are distributed randomly. Preservation regions around each WN played key roles in improving the performance of covert communication in the network while allowing a vanishing fraction of LNs that are in outage. We utilized and modified three existing network communication schemes so that they are amenable in the covert communication setting. We showed that each proposed scheme can achieve asymptotically same throughput scaling that the original scheme can achieve if the proposed and original schemes are performed under the same average transmit power constraint. We observed that the proposed schemes asymptotically achieve upper bounds on the throughput scaling that are derived under the assumption that every active LN uses the same average transmit power over an arbitrary window of $l$ channel uses. Although this assumption violates the generality of the converse theorems, it is not unreasonable to believe that in some network scenarios, this assumption is valid. Our results can be extended to a network of arbitrary size by simply adapting the size of a preservation region.   

We present some possible future directions for research in the following:
\begin{itemize}
\item Deriving a general converse theorem without the assumption on the transmit power of LNs in Theorems \ref{thm:conv1} and \ref{thm:conv2}.
\item In this paper, we assume that the WNs are non-communicating. What is the effect of possible cooperation between the WNs? How will the covertness constraint be further strengthened? In addition, what are the similarities between the cooperation between the WNs and the interferences suffered by the LNs? 
\item We assumed that the number of LNs is larger than the number of WNs. If this is not true, what is the capacity scaling in this new class of networks? In this case, we conjecture that the use of preservation regions may not be as effective as in this paper. 
\end{itemize}

\appendices

\section{Proof of \eqref{eqn:square}} \label{app:square}
To prove \eqref{eqn:square}, we show that the relative entropy between two non-negative random variables is invariant when each of the random variables is squared. Let $Z$ and $N$ be two non-negative random variables. In addition, let $f_{(\cdot)}$ and $F_{(\cdot)}$ be the density function and the cumulative distribution function (CDF) of random variable $(\cdot)$, respectively. Then, for any $a \geq 0$, we can relate the CDFs of $Z$ and $Z^2$ as 
\begin{align}
F_{Z^2}(a) = \P(Z^2 \leq a) = \P(Z \leq \sqrt{a}) = F_
Z(\sqrt{a}).
\end{align}  
Thus, we have
\begin{align}
f_{Z^2}(a) = \frac{\mathrm{d}}{\mathrm{d}a}F_{Z^2}(a) = \frac{\mathrm{d}}{\mathrm{d}a}F_{Z}(\sqrt{a}) = \frac{1}{2\sqrt{a}}f_{Z}(\sqrt{a}).
\end{align}  
Substituting this into the formula for the relative entropy, we have
\begin{align}
D(f_{Z^2} \| f_{N^2}) &= \int_{0}^{\infty}f_{Z^2}(z)\log \frac{f_{Z^2}(z)}{f_{N^2}(z)}\, \mathrm{d}z \\
&= \int_{0}^{\infty} \frac{1}{2\sqrt{z}} f_{Z}(\sqrt{z})\log \frac{\frac{1}{2\sqrt{z}}f_{Z}(\sqrt{z})}{\frac{1}{2\sqrt{z}}f_{N}(\sqrt{z})}\, \mathrm{d}z \\
&= \int_{0}^{\infty} \frac{1}{2\sqrt{z}} f_{Z}(\sqrt{z})\log \frac{f_{Z}(\sqrt{z})}{f_{N}(\sqrt{z})}\, \mathrm{d}z.
\end{align}  
Now let $t = \sqrt{z}$. Then $\mathrm{d}z = 2t \, \mathrm{d}t$ and the range of integration is unchanged. Thus,
\begin{align}
D(f_{Z^2} \| f_{N^2}) &= \int_{0}^{\infty} \frac{1}{2t} f_{Z}(t)\log \frac{f_{Z}(t)}{f_{N}(t)}2t \, \mathrm{d}t \\
&= \int_{0}^{\infty}  f_{Z}(t)\log \frac{f_{Z}(t)}{f_{N}(t)} \, \mathrm{d}t \\
&= D(f_Z \| f_N),
\end{align} 
which proves \eqref{eqn:square}.

\section{Vanishing Outage Fraction} \label{app:outage}
We prove that in the three proposed schemes, the outage fraction vanishes if $\gamma = \frac{\kappa}{2} + \epsilon$ for any $\epsilon > 0$. First consider the detoured MH scheme. In the detoured MH scheme, there are two cases of LNs in outage: the set of LN pairs that the LSN or the LDN in each pair belongs to a expanded preservation region (denote this set by $T_1$) and the set of LN pairs such that there is no possible data path between the LSN and the LDN of each pair (denote this set by $T_2$). Then, we have to show that $\frac{|T_1 \cup T_2|}{n}$ tends to zero as $n$ tends to infinity. Clearly, we have
\begin{align} \label{eqn:union}  
\frac{|T_1 \cup T_2|}{n} \leq \frac{|T_1|}{n} + \frac{|T_2|}{n},
\end{align}
and we will instead show that each of the two terms in the right hand side above converges to zero. 

Let us consider the first term. We use a simple fact that a preservation region of area $n^{-\kappa-2\epsilon}\log n$ can result in at most $c_1n(n^{-\kappa-2\epsilon}\log n)$ outage LSNs where $c_1$ is a constant independent of $\kappa$ and $n$. Then, we have
\begin{align}  
\frac{|T_1|}{n} \leq \frac{2c_1n^{\kappa}n(n^{-\kappa-2\epsilon}\log n)}{n} = \frac{2c_1 \log n}{n^{2\epsilon}}.
\end{align}
The last term tends to zero as $n$ tends to infinity. 

\begin{figure}
 \centering
  {
  \includegraphics[width=45mm]{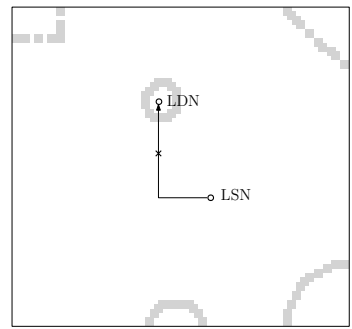}}
  \caption{Examples of expanded preservation regions that block the data paths of some LN pairs.} \label{fig:outagebound}
\end{figure}
Next consider the second term in \eqref{eqn:union}. Fig. \ref{fig:outagebound} shows several cases of expanded preservation regions that block the data paths of some LN pairs and thus result in the extra outage fraction. To upper bound the second term, let us focus on the case of the expanded preservation region that has the highest ratio of the number of induced outage LNs per the number of preservation regions in the expanded preservation regions. This case has the following properties: The expanded preservation 
\begin{itemize}  
\item Surrounds the corner of the network.
\item Contains $4\kappa \log n$ preservation regions (upper bound on the number of preservation regions w.h.p.) because the area of a region enclosed by a certain shape is proportional to the square of the perimeter of the shape. 
\item Has shape of circle because the area of a region is maximized by a circle when a perimeter is given.  
\item The preservation regions in the expanded preservation region do not overlap with each other and maintain a maximum interval (asymptotically $b$) each other. 
\end{itemize}
We now calculate the area of the region surrounded by the expanded preservation region that has the properties described above. In this case, the region has the shape of quartile circle, and the perimeter of the quartile circle is less than $c_2(4\kappa \log n)(\sqrt{2}n^{-\frac{\kappa}{2}-\epsilon}\sqrt{\log n})$, where $c_2$ is a constant independent of $\kappa$ and $n$. Then, the area of the quartile circle is less than $[c_2(4\kappa \log n)(\sqrt{2}n^{-\frac{\kappa}{2}-\epsilon}\sqrt{\log n})]^2 / \pi = c_3 (\log n)^3 n^{-\kappa-2\epsilon}$, and there are at most $2c_3 (\log n)^3 n^{1-\kappa-2\epsilon}$ LNs w.h.p., where $c_3$ is a constant independent of $\kappa$ and $n$. If we assume that there are $\frac{n^{\kappa}}{4\kappa \log n}$ expanded preservation regions of this case, we have
\begin{align}  
\frac{|T_2|}{n} \leq \frac{2c_3 (\log n)^3 n^{1-\kappa-2\epsilon}n^{\kappa}}{4\kappa n\log n} = \frac{c_3 (\log n)^2}{2\kappa n^{2\epsilon}}.
\end{align} 
The last term tends to zero as $n$ tends to infinity. This proves that \eqref{eqn:union} vanishes. 

The case of the modified HC scheme is a trivial case because the set of outage LNs is the set of LN pairs that the LSN or the LDN in each pair is in a preservation region. 
Finally, in the detoured hybrid scheme, the proof can be completed by following the same steps in the proof for the case of the detoured MH scheme.

\bibliographystyle{IEEEtran}
\bibliography{Covert_scaling_ChoLeeTan}

\end{document}